\documentclass[11pt]{article}

\long\def\comment#1{\ifdim\overfullrule>0pt{\sf[{#1}]}\fi}

\usepackage{comment,amsmath,amsthm,graphicx}
\usepackage{latexsym}
\usepackage{epsfig}
\usepackage{bbm}
\usepackage[usenames]{color}
\usepackage{colordvi}
\usepackage{pdfsync}
\usepackage{amsfonts}
\usepackage{algorithm,algorithmicx,algpseudocode}
\usepackage{mathrsfs}

\usepackage{graphicx}
\usepackage{tikz}
\usetikzlibrary{arrows,calc,patterns}
\usetikzlibrary{mindmap,trees}
\usetikzlibrary{plotmarks}
\usetikzlibrary{decorations.pathreplacing}
\usepackage{color}
\usepackage{appendix}

\newtheorem{theorem}{Theorem}
\newtheorem{definition}{Definition}

\newtheorem{claim}{Claim} \newtheorem{lemma}{Lemma}

\ifx\pdftexversion\undefined
\usepackage[colorlinks,linkcolor=black,filecolor=black,citecolor=black,urlco
lor=black,pdfstartview=FitH]{hyperref}
\else
\usepackage[colorlinks,linkcolor=blue,filecolor=blue,citecolor=blue,urlcolor
=blue,pdfstartview=FitH]{hyperref}
\fi

\newcommand{\namedref}[2]{\hyperref[#2]{#1~\ref*{#2}}}

\newcommand{\alert}[1]{\textbf{\color{red}
[[[#1]]]}\marginpar{\textbf{\color{red}**}}\typeout{ALERT:
\the\inputlineno: #1}}

\begin{document} \bibliographystyle{alpha}
\def\proofend{\hfill$\Box$\medskip}
\def\Proof{\noindent{\bf Proof:\ \ }}
\def\Sketch{\noindent{\bf Sketch:\ \ }}
\def\eps{\epsilon}

\title{The Alternating Stock Size Problem and the Gasoline
  Puzzle}

\author{Alantha Newman\thanks{
CNRS-Universit\'e Grenoble Alpes and G-SCOP.
Supported in part by LabEx PERSYVAL-Lab (ANR--11-LABX-0025).
 {{\tt alantha.newman@grenoble-inp.fr}}} \and Heiko
  R\"oglin\thanks{Universit\"at Bonn. Supported by ERC Starting Grant
    306465 (BeyondWorstCase). {\tt roeglin@cs.uni-bonn.de}} \and
  Johanna Seif\thanks{Ecole Normale Sup\'erieure de Lyon. {\tt johanna.seif@ens-lyon.fr}} }

\maketitle

\begin{abstract}
Given a set $S$ of integers whose sum is zero, consider the
problem of finding a permutation of these integers such that
(i) all prefix sums of the ordering are nonnegative and
(ii) the maximum value of a prefix sum is minimized.  
Kellerer et al.\ referred to this problem as the {\em stock size problem} and
showed that it can be approximated to within 3/2.  They also
showed that an approximation ratio of 2 can be achieved via several
simple algorithms.

We consider a related problem, which we call the {\em alternating stock
size problem}, where the numbers of positive and negative integers in
the input set $S$ are equal.  The problem is the same as above, but we
are additionally required to alternate the positive and negative
numbers in the output ordering.  This problem also has several simple
2-approximations.  We show that it can be approximated to within 1.79.

Then we show that this problem is closely related to an optimization
version of the gasoline puzzle due to Lov\'asz, in which we want to
minimize the size of the gas tank necessary to go around the track.
We present a 2-approximation for this problem, using a
natural linear programming relaxation whose
feasible solutions are doubly stochastic matrices.  
Our novel rounding algorithm is based on a transformation that
yields another doubly stochastic matrix with special properties, 
from which we can extract a suitable permutation.



\end{abstract}


\section{Introduction}

Suppose there is a set of jobs that can be processed in any order.
Each job requires a specified amount of a particular resource,
such as gasoline, which can be supplied in an amount chosen from a
specified set of quantities.  The limitation is that the storage space
for this resource is bounded, so it must be replenished as it is used.
The goal is to order the jobs and the replenishment amounts so that
the required quantity of the resource is always available for the job
being processed and so that the storage space is never exceeded.

More formally, we are given a set of integers $Z = \{z_1, z_2, \dots
z_n\}$ whose sum is zero.  For a permutation $\sigma$, a prefix sum is
$\sum_{i=1}^t z_{\sigma(i)}$ for $t \in [1,n]$.  Our goal is to find a
permutation of the elements in $Z$ such that (i) each prefix sum is
nonnegative and (ii) the maximum prefix sum is minimized.  (Placing
the elements with positive values in front of the elements with
negative values satisfies (i) and therefore yields a
feasible---although possibly far from optimal---solution.)  This
problem is known as the {\em stock size problem}.  Kellerer, Kotov,
Rendl and Woeginger presented a simple algorithm with a guarantee of
$\mu_x + \mu_y$, where $\mu_x$ is the largest number in $Z$, and
$\mu_y$ is the absolute value of the negative number with the largest
absolute value in $Z$.  (We sometimes use $\mu = \max\{\mu_x,
\mu_y\}$.)  Since both $\mu_x$ and $\mu_y$ are lower bounds on the
value~$S^*$ of an optimal solution, this shows that the problem can be
approximated to within a factor of 2.  Additionally, they presented
algorithms with approximation guarantees of $8/5$ and
$3/2$~\cite{kellerer1998stock}.

\subsection{The Alternating Stock Size
  Problem}\label{subsec:AlternatingSSProblem}

In this paper, we first consider a restricted version of the stock size
problem in which we require that the positive and negative numbers in
the output permutation alternate.  We call this problem the
{\em alternating stock size problem}.  A motivation for this problem
is that it would allow for task scheduling in advance of knowing the
input data.  For example, suppose we want to stock and remove items
from a warehouse and each task will occupy a time slot.
If we
want to plan ahead, we may want to designate each slot as a stocking
or a removing slot in advance--for example, all odd time (night) slots will be used
for stocking and all even time (day) slots for destocking.  This could be
beneficial in situations where some preparation is required for each
type of time slot.

The input for our new problem is two sets of positive integers, $X =
\{x_1 \geq \dots \geq x_n\}$ and $Y= \{y_1 \geq \dots \geq y_n\}$,
such that $|X|=|Y|$, and the two sets have equal sums.  The elements
of $X$ represent the elements to be ``added'' and the elements of $Y$
are those to be ``removed''.  Note that here, $\mu_y = y_1$
and $\mu_x = x_1$.  We now formally define the new problem.
\begin{definition}
The goal of the {\em{alternating stock size problem}} is to find permutations
$\sigma$ and $\nu$ such that
\begin{enumerate}

\item[(i)] for $t \in [1,n]$, $\sum_{i=1}^t x_{\sigma(i)} - y_{\nu(i)}
  \geq 0$,

\item[(ii)] $\max\limits_{1 \leq t \leq n} \sum_{i=1}^t (x_{\sigma(i)}
  - y_{\nu(i-1)})$
is minimized, where $y_{\nu(0)} = 0$.

\end{enumerate}
\end{definition}

Although this problem is a variant of the stock size problem, the
algorithms found in \cite{kellerer1998stock} do not provide
approximation guarantees, since they do not necessarily produce
feasible solutions for the alternating problem.  Indeed, even the
optimal solutions for these two problems on the same instance can
differ greatly.  The following example illustrates this:
\begin{eqnarray*}
X & = & \{\underbrace{p-1, \dots, p-1}_{p \text{ entries }},2,
\underbrace{1, \dots, 1}_{p(p-1) \text{ entries }}\},\\
Y & = & \{\underbrace{p, \dots, p}_{p-1 \text{ entries
}},\underbrace{1, 1, 1, \dots, 1}_{p(p-1)+2 \text{ entries
}}\}.
\end{eqnarray*}
For this instance, the optimal value for the alternating problem is at
least $2p-3$, while it is $p$ for the original stock size problem.
Thus, this example exhibits a gap arbitrarily close to 2 between the
optimal solutions for the two problems.

We can show the following facts about the alternating problem.  First, 
there is always a feasible solution.  Second, the problem is NP-hard (as
is the stock size problem).  And third, it is still the case that $2\mu$ is
an upper bound on the value of an optimal solution.  Our main result
for this problem is to give an algorithm with an approximation
guarantee of $1.79$ in Section \ref{sec:alternating}.

\subsection{The Gasoline Problem}\label{subsec:GasolineProblem}

The following well-known puzzle appears on page 31 in \cite{lovasz}:

\begin{quote}
Along a speed track there are some gas stations.  The total amount of
gasoline available in them is equal to what our car (which has a very
large tank) needs for going around the track.  Prove that there is a
gas station such that if we start there with an empty tank, we shall
be able to go around the track without running out of gasoline.
\end{quote}

Suppose that the capacity of each gas station is represented by a
positive integer and the distance of each road segment is represented
by a negative integer.  For simplicity, suppose that it takes one unit
of gas to travel one unit of road.  Then the assumption of the puzzle
implies that the sum of the positive integers equals the absolute
value of the sum of the negative integers.  In fact, if we are allowed
to permute the gas stations and the road segments (placing exactly one
gas station between every pair of consecutive road segments), and our
goal is to minimize the size of the gas tank required to go around the
track (beginning from a feasible starting point), then this is
exactly the alternating stock size problem.

This leads to the following natural problem: Suppose the road segments
are fixed and we are only allowed to rearrange (i.e.\ permute) the gas
stations.  In other words, between each pair of consecutive road
segments (represented by negative integers), there is a spot for
exactly one gas station (represented by positive integers, the
capacities), and we can choose which gas station to place in each
spot.  The goal is to minimize the size of the tank required to get
around the track, assuming we can choose our starting gas station.
What is the complexity of this problem?

We argue in Appendix \ref{sec:np} that this problem is NP-hard.  Our
algorithm for the alternating stock size problem specifically requires
that there is flexibility in placing both the $x$-values and the
$y$-values.  Therefore, it does not appear to be applicable to this
problem, where the $y$-values are pre-assigned to fixed positions.
Let us now formally define the {\em gasoline
  problem}, which is the second problem we will consider in this paper.

As input, we are given the two sets of positive integers $X = \{x_1 \geq
x_2 \geq \dots \geq x_n\}$ and $Y = \{y_1, y_2, \dots, y_n\}$,
where the $y_i$'s are fixed in the given order and $\sum_{i=1}^n x_i =
\sum_{i=1}^n y_i$.  Our goal is to find a permutation $\pi$ that
minimizes the value of $\eta$:
\begin{eqnarray}
\forall [k,\ell]: \quad \Bigg|\sum_{i\in[k,\ell]} x_{\pi(i)} -
\sum_{i\in[k,\ell-1]}y_{i}\Bigg| & \leq & \eta.\label{main_objective}
\end{eqnarray}
Given a circle with $n$ points labeled $1$ through $n$, the interval
$[k,\ell]$ denotes a consecutive subset of integers assigned to points
$k$ through $\ell$.  For example, $[5,8] = \{5,6,7,8\}$, and $[n-1,3]
= \{n-1,n,1,2,3\}$.  We will often use $\mu_x$ to refer to $x_1$,
i.e.\ the maximum $x$-value, which is a lower bound on the optimal
value of a solution.

Observe that in~\eqref{main_objective} we consider only intervals that
contain one more $x$-value than $y$-value.  One might argue that, in
order to model our problem correctly, one also has to look at
intervals that contain one more $y$-value than $x$-value. However,
let~$I$ be such an interval and let~$I'=[1,n]\setminus I$. Then the
absolute value of the difference of the $x$-values and the $y$-values
is the same in~$I$ and~$I'$ (with inverted signs) due to the
assumption~$\sum_{i=1}^n x_i = \sum_{i=1}^n y_i$.

We can also write the constraint \eqref{main_objective} as:
\begin{eqnarray}
\forall k:  \sum_{i\in[1,k]} x_{\pi(i)} -
\sum_{i\in[1,k-1]}y_{i} & \leq & \beta, \label{main_obj_beta}\\
\forall k: \sum_{i\in[1,k]} x_{\pi(i)} - \sum_{i\in[1,k]}y_{i} & \geq &
\alpha, \label{main_obj_alpha}
\end{eqnarray}
where $\alpha \leq 0$, $\beta \geq 0$ and $\eta = \beta-\alpha$.  This
version is slightly more general since it encompasses the scenario
where we would like to minimize $\beta$ for some fixed value of
$\alpha$.  (With these constraints, it is no longer required that the
sum of the $x_i$'s equals the sum of the $y_i$'s.)

What is the approximability of this problem?  Getting a constant
factor approximation appears to be a challenge since the following
example shows that it is no longer the case that $2\mu$ is an upper
bound.  Despite this, we show in Section \ref{section:GasolineProblem}
that there is in fact a 2-approximation algorithm for the gasoline
problem.

\paragraph{Example showing unbounded gap between $OPT$ and
  $\mu$.}

Suppose $X$ and $Y$ each have the following $n$ entries:
\begin{eqnarray*}
X  =  \{\underbrace{1, 1, \dots, 1, 1, 1, \dots, 1}_{n \text{ entries
}}\}, \quad
Y  =  \{\underbrace{2, 2, \dots, 2}_{\frac{n}{2} \text{ entries }}, \underbrace{0,0,
\dots, 0}_{\frac{n}{2} \text{ entries }}\}.
\end{eqnarray*}
In the preceding example, $\mu = 2$.  However, the optimal value is
$n/2$.

\subsection{Generalizations of the Gasoline Problem}
\label{subsec:SlatedSSProblem}

The requirement that the $x$- and $y$-jobs alternate may seem to be
somewhat artificial or restrictive.  A natural generalization of the
gasoline problem (which we will refer to as the {\em generalized
  gasoline problem}) is where the $y$-jobs are assigned to a set of
predetermined positions, which are not necessarily alternating.  As in
the gasoline problem, our goal is to assign the $x$-jobs to the
remaining slots so as to minimize the difference between the maximum
and the minimum prefix.  There is a simple reduction from this
seemingly more general problem to the gasoline problem.  Let $X=
\{x_1 \geq x_2 \geq \dots \geq x_{n_x}\}$ and $Y = \{y_1, y_2, \dots,
y_{n_y}\}$ be the input, where the $y$-jobs are assigned to $n_y$
(arbitrary) slots.  The remaining $n_x$ slots are for the $x$-jobs.
To reduce to an instance of the gasoline problem (with alternation),
we do the following.  For each set of $y$-jobs assigned to adjacent
slots, we add them up to form a single job in a single slot.  For each
pair of consecutive $x$-slots, we place a new $y$-slot between them
where the assigned $y$-job has value zero.  Thus, we obtain an
instance of the gasoline problem as originally defined in the
beginning of this section.

Our new algorithm, developed in Section \ref{section:GasolineProblem}
to solve the gasoline problem, can also be applied to a natural
generalization of the alternating stock size problem, in which we
relax the required alternation between the $x$- and $y$-jobs and
consider a scenario in which each slot is labeled as an $x$- or a
$y$-slot and can only accomodate a job of the designated type.  In
other words, in the solution, the $x$-jobs and $y$-jobs will follow
some specified pattern that is not necessarily alternating.  The goal
is to find a feasible assignment of $x$- and $y$-jobs to $x$- and
$y$-slots, respectively, that minimizes the difference between the
prefixes with highest and lowest values.  Since this is simply a
generalization of the stock size problem with the additional condition
that each slot is {\em slated} as an $x$- or a $y$-slot, we refer to
this problem as the {\em slated stock size problem}.

Formally, we are given two sets of positive integers $X = \{x_1 \geq
x_2 \geq \dots \geq x_{n_x}\}$ and $Y = \{y_1 \geq y_2 \geq \dots \geq
y_{n_y}\}$, and $n = n_x + n_y$ slots, each designated as either an
$x$-slot or a $y$-slot.  Let $I_x$ and $I_y$ denote the indices of the
$x$- and $y$-slots, respectively, and let $P$ denote a prefix.  Then,
the objective is to find a permutation $\pi$ that minimizes the value
of $\beta-\alpha$, where
\begin{eqnarray}
\forall P: \quad \alpha ~ \leq \sum_{i\in P\cap I_x} x_{\pi(i)} - \sum_{i\in P\cap I_y}y_{\pi(i)} ~\leq ~\beta.\label{main_obj_alpha2}
\end{eqnarray}
For this problem, we obtain an algorithm that computes a solution with
value at most $OPT + \mu_x + \mu_y \leq 3\ OPT$.

\subsection{Related Work}

The work most related to the alternating stock size problem is
contained in the aforementioned paper by Kellerer et
al.~\cite{kellerer1998stock}.  Earlier, Abdel-Wahab and Kameda studied
a variant of the stock size problem in which the output sequence of
the jobs is required to obey a given set of precedence constraints,
but the stock size is also allowed to be negative.  They gave a
polynomial-time algorithm for the case when the precedence constraints
are series parallel~\cite{abdel1978scheduling}.  The gasoline problem
and its generalization are related to those found in a widely-studied
research area known as resource constrained scheduling, where the goal
is usually to minimize the completion time or to maximize the number
of jobs completed in a given timeframe while subject to some limited
resources~\cite{blazewicz1983scheduling,carlier1982scheduling}.  For
example, in addition to time on a machine, a job could require a
certain amount of another resource and would be eligible to be
scheduled only if the {\em inventory} for this resource is sufficient.

A general framework for these types of problems is called scheduling
with nonrenewable resources.  Here, {\em nonrenewable} means not
abundantly available, but rather replenished according to some rules,
such as periodically and in predetermined increments (as in the
gasoline problem), or in specified increments that can be scheduled by
the user (as in the alternating stock size problem), or at some
arbitrary fixed timepoints.  Examples for scheduling problems in this
framework are described by Briskorn et
al.~\cite{briskorn2010complexity}, by Gy\"orgyi and
Kis~\cite{gyorgyi2014approximation, gyorgyi2015approximability}, and
by Morsy and Pesch~\cite{morsy2015approximation}.  Although the
admissibility of a schedule is affected by the availability of a
resource (e.g. whether or not there is sufficient inventory),
minimizing the inventory is not a main objective in these works.

For example, suppose we are given a set of jobs to be scheduled on a
single machine.  Each job consumes some resource and is only allowed
to be scheduled at a timepoint if there is a sufficient supply
available for that job at this timepoint.  Jobs may have different
resource requirements.  Periodically, at timepoints and in increments
known in advance, the resource will be replenished.  The goal is to
minimize the completion time.  If at some timepoint there is
insufficient inventory for any job to be scheduled, then no job can be
run, leading to gaps in the schedule and ultimately a later completion
time.  This problem of minimizing the completion time is polynomial
time solvable (sort the jobs according to resource requirement), but
an optimal schedule may contain idle times.

Suppose that we have some investment amount $\alpha$ that we can add
to the inventory in advance to ensure that there is always sufficient
inventory to schedule some job, resulting in a schedule with no empty
timeslots, i.e. the optimal completion time.  There is a natural
connection between this scenario and the gasoline problem: Let
$|\alpha|$ in Equation \eqref{main_obj_alpha} denote the available
investment.  For this investment, suppose we wish to minimize $\beta$,
which is the maximum inventory, to complete the jobs in the
optimal completion time.  For any feasible $\alpha$ and $\beta$, our
algorithm in Section \ref{section:GasolineProblem} produces a schedule
with the optimal completion time using inventory size at most $\beta +
\mu$.

There are other works that directly address the problem of minimizing
the maximum or cumulative inventory.  Monma considers a problem in
which each job has a specified effect on the inventory
level~\cite{monma1980sequencing}.  Neumann and Schwindt consider a
scheduling problem in which the inventory is subject to both upper and
lower bounds~\cite{neumann2003project}.  However, to the best of our
knowledge, our work is the first to give approximation algorithms for
the problem of minimizing the maximum inventory for nonrenewable
resource scheduling with fixed replenishments.

Finally, we note that the stock size problem is closely related to the
Steinitz problem in one dimension.  Given a set of vectors $v_1, v_2,
\dots v_n \in \mathbb{R}^d$ where $||v_i|| \leq 1$ for some fixed norm
and $\sum_{i=1}^n v_i = 0$, the Steinitz problem is to find a
permutation of the vectors so that the norm of the sum of each prefix
is bounded.  The objective of the Steinitz problem is to give a
worst-case bound in terms of $d$ on the value of a maximum prefix over
all inputs.  The objective of the stock size problem, however, is to provide a relative bound on the value of the maximum
prefix for a specific input instance.  Via the bound of $d$ for
the Steinitz problem~\cite{grinberg1980value}, we can obtain a
2-approximation algorithm for the stock size problem, but this does
not match the best-known bound of $3/2$ due to Kellerer et
al.~\cite{kellerer1998stock}.  We refer the reader to Section 6 of
\cite{kellerer1998stock} for a discussion of the connection between
these two problems.  For more on the low-dimensional Steinitz problem,
we refer the reader to the work of Banaszczyk~\cite{banaszczyk1987steinitz}.


\section{Algorithms for the Alternating Stock Size
  Problem} \label{sec:alternating}

The existence of a feasible solution for the alternating stock size
problem follows from the solution for the gasoline puzzle.  (See
Figure \ref{drawing_proof} for more details.)
\begin{figure}
  \centering
  \scalebox{.2}{
\ifx\du\undefined
  \newlength{\du}
\fi
\setlength{\du}{15\unitlength}
\begin{tikzpicture}
\pgftransformxscale{1.000000}
\pgftransformyscale{-1.000000}
\definecolor{dialinecolor}{rgb}{0.000000, 0.000000, 0.000000}
\pgfsetstrokecolor{dialinecolor}
\definecolor{dialinecolor}{rgb}{1.000000, 1.000000, 1.000000}
\pgfsetfillcolor{dialinecolor}
\definecolor{dialinecolor}{rgb}{0.921569, 0.000000, 0.000000}
\pgfsetstrokecolor{dialinecolor}
\node[anchor=west] at (12.179300\du,-5.524740\du){};
\pgfsetlinewidth{0.150000\du}
\pgfsetdash{}{0pt}
\pgfsetdash{}{0pt}
\pgfsetbuttcap
{
\definecolor{dialinecolor}{rgb}{0.000000, 0.000000, 0.000000}
\pgfsetfillcolor{dialinecolor}
\definecolor{dialinecolor}{rgb}{0.000000, 0.000000, 0.000000}
\pgfsetstrokecolor{dialinecolor}
\draw (-0.129681\du,0.000000\du)--(-0.129681\du,0.000000\du);
}
\pgfsetlinewidth{0.200000\du}
\pgfsetdash{}{0pt}
\pgfsetdash{}{0pt}
\pgfsetbuttcap
{
\definecolor{dialinecolor}{rgb}{0.000000, 0.000000, 0.000000}
\pgfsetfillcolor{dialinecolor}
\pgfsetarrowsend{to}
\definecolor{dialinecolor}{rgb}{0.000000, 0.000000, 0.000000}
\pgfsetstrokecolor{dialinecolor}
\draw (-0.982925\du,0.089815\du)--(38.146800\du,0.089815\du);
}
\pgfsetlinewidth{0.200000\du}
\pgfsetdash{}{0pt}
\pgfsetdash{}{0pt}
\pgfsetbuttcap
{
\definecolor{dialinecolor}{rgb}{0.000000, 0.000000, 0.000000}
\pgfsetfillcolor{dialinecolor}
\pgfsetarrowsend{to}
\definecolor{dialinecolor}{rgb}{0.000000, 0.000000, 0.000000}
\pgfsetstrokecolor{dialinecolor}
\draw (0.005458\du,1.092100\du)--(0.015412\du,-17.262100\du);
}
\pgfsetlinewidth{0.200000\du}
\pgfsetdash{}{0pt}
\pgfsetdash{}{0pt}
\pgfsetmiterjoin
\pgfsetbuttcap
{
\definecolor{dialinecolor}{rgb}{0.074510, 0.023529, 0.843137}
\pgfsetfillcolor{dialinecolor}
{\pgfsetcornersarced{\pgfpoint{0.000000\du}{0.000000\du}}\definecolor{dialinecolor}{rgb}{0.074510, 0.023529, 0.843137}
\pgfsetstrokecolor{dialinecolor}
\draw (-0.129681\du,0.000000\du)--(3.129680\du,-9.000000\du)--(6.129680\du,-7.000000\du)--(9.129680\du,-11.000000\du)--(12.000000\du,1.000000\du)--(15.129700\du,-13.000000\du)--(18.129700\du,-9.000000\du)--(21.000000\du,-11.000000\du)--(25.000000\du,5.000000\du);
}}
\pgfsetlinewidth{0.200000\du}
\pgfsetdash{}{0pt}
\pgfsetdash{}{0pt}
\pgfsetbuttcap
{
\definecolor{dialinecolor}{rgb}{0.000000, 0.000000, 0.000000}
\pgfsetfillcolor{dialinecolor}
\definecolor{dialinecolor}{rgb}{0.000000, 0.000000, 0.000000}
\pgfsetstrokecolor{dialinecolor}
\draw (-0.039866\du,-0.134723\du)--(0.015717\du,5.219620\du);
}
\pgfsetlinewidth{0.200000\du}
\pgfsetdash{}{0pt}
\pgfsetdash{}{0pt}
\pgfsetmiterjoin
\pgfsetbuttcap
{
\definecolor{dialinecolor}{rgb}{0.074510, 0.023529, 0.843137}
\pgfsetfillcolor{dialinecolor}
{\pgfsetcornersarced{\pgfpoint{0.000000\du}{0.000000\du}}\definecolor{dialinecolor}{rgb}{0.074510, 0.023529, 0.843137}
\pgfsetstrokecolor{dialinecolor}
\draw (25.000000\du,5.000000\du)--(28.000000\du,-14.000000\du)--(31.000000\du,-5.000000\du)--(32.000000\du,-6.000000\du)--(37.000000\du,0.000000\du);
}}
\pgfsetlinewidth{0.150000\du}
\pgfsetdash{{1.000000\du}{1.000000\du}}{0\du}
\pgfsetdash{{1.000000\du}{1.000000\du}}{0\du}
\pgfsetbuttcap
{
\definecolor{dialinecolor}{rgb}{1.000000, 0.027451, 0.027451}
\pgfsetfillcolor{dialinecolor}
\definecolor{dialinecolor}{rgb}{1.000000, 0.027451, 0.027451}
\pgfsetstrokecolor{dialinecolor}
\draw (25.000000\du,-15.000000\du)--(25.000000\du,6.000000\du);
}
\pgfsetlinewidth{0.200000\du}
\pgfsetdash{}{0pt}
\pgfsetdash{}{0pt}
\pgfsetmiterjoin
\pgfsetbuttcap
{
\definecolor{dialinecolor}{rgb}{1.000000, 0.027451, 0.027451}
\pgfsetfillcolor{dialinecolor}
{\pgfsetcornersarced{\pgfpoint{0.000000\du}{0.000000\du}}\definecolor{dialinecolor}{rgb}{1.000000, 0.027451, 0.027451}
\pgfsetstrokecolor{dialinecolor}
\draw (0.941995\du,2.132950\du)--(2.084740\du,0.990209\du)--(0.184504\du,0.251514\du)--(0.976624\du,2.063690\du);
}}
\pgfsetlinewidth{0.200000\du}
\pgfsetdash{{1.000000\du}{1.000000\du}}{0\du}
\pgfsetdash{{1.000000\du}{1.000000\du}}{0\du}
\pgfsetmiterjoin
\pgfsetbuttcap
{
\definecolor{dialinecolor}{rgb}{1.000000, 0.027451, 0.027451}
\pgfsetfillcolor{dialinecolor}
\definecolor{dialinecolor}{rgb}{1.000000, 0.027451, 0.027451}
\pgfsetstrokecolor{dialinecolor}
\pgfpathmoveto{\pgfpoint{1.565310\du}{1.544270\du}}
\pgfpathcurveto{\pgfpoint{4.653850\du}{4.673990\du}}{\pgfpoint{16.892800\du}{6.796550\du}}{\pgfpoint{25.005400\du}{5.025790\du}}
\pgfusepath{stroke}
}
\end{tikzpicture}} \hspace{10mm}
  \scalebox{.2}{
\ifx\du\undefined
  \newlength{\du}
\fi
\setlength{\du}{15\unitlength}
\begin{tikzpicture}
\pgftransformxscale{1.000000}
\pgftransformyscale{-1.000000}
\definecolor{dialinecolor}{rgb}{0.000000, 0.000000, 0.000000}
\pgfsetstrokecolor{dialinecolor}
\definecolor{dialinecolor}{rgb}{1.000000, 1.000000, 1.000000}
\pgfsetfillcolor{dialinecolor}
\definecolor{dialinecolor}{rgb}{0.921569, 0.000000, 0.000000}
\pgfsetstrokecolor{dialinecolor}
\node[anchor=west] at (12.176700\du,25.519000\du){};
\pgfsetlinewidth{0.150000\du}
\pgfsetdash{}{0pt}
\pgfsetdash{}{0pt}
\pgfsetbuttcap
{
\definecolor{dialinecolor}{rgb}{0.000000, 0.000000, 0.000000}
\pgfsetfillcolor{dialinecolor}
\definecolor{dialinecolor}{rgb}{0.000000, 0.000000, 0.000000}
\pgfsetstrokecolor{dialinecolor}
\draw (-0.132319\du,31.043800\du)--(-0.132319\du,31.043800\du);
}
\pgfsetlinewidth{0.200000\du}
\pgfsetdash{}{0pt}
\pgfsetdash{}{0pt}
\pgfsetbuttcap
{
\definecolor{dialinecolor}{rgb}{0.000000, 0.000000, 0.000000}
\pgfsetfillcolor{dialinecolor}
\pgfsetarrowsend{to}
\definecolor{dialinecolor}{rgb}{0.000000, 0.000000, 0.000000}
\pgfsetstrokecolor{dialinecolor}
\draw (-0.950200\du,31.043800\du)--(38.179500\du,31.043800\du);
}
\pgfsetlinewidth{0.200000\du}
\pgfsetdash{}{0pt}
\pgfsetdash{}{0pt}
\pgfsetbuttcap
{
\definecolor{dialinecolor}{rgb}{0.000000, 0.000000, 0.000000}
\pgfsetfillcolor{dialinecolor}
\pgfsetarrowsend{to}
\definecolor{dialinecolor}{rgb}{0.000000, 0.000000, 0.000000}
\pgfsetstrokecolor{dialinecolor}
\draw (-0.082319\du,31.993800\du)--(-0.075948\du,6.280030\du);
}
\pgfsetlinewidth{0.150000\du}
\pgfsetdash{{1.000000\du}{1.000000\du}}{0\du}
\pgfsetdash{{1.000000\du}{1.000000\du}}{0\du}
\pgfsetbuttcap
{
\definecolor{dialinecolor}{rgb}{1.000000, 0.027451, 0.027451}
\pgfsetfillcolor{dialinecolor}
\definecolor{dialinecolor}{rgb}{1.000000, 0.027451, 0.027451}
\pgfsetstrokecolor{dialinecolor}
\draw (12.090600\du,10.937100\du)--(12.090600\du,31.937100\du);
}
\pgfsetlinewidth{0.200000\du}
\pgfsetdash{}{0pt}
\pgfsetdash{}{0pt}
\pgfsetmiterjoin
\pgfsetbuttcap
{
\definecolor{dialinecolor}{rgb}{0.074510, 0.023529, 0.843137}
\pgfsetfillcolor{dialinecolor}
{\pgfsetcornersarced{\pgfpoint{0.000000\du}{0.000000\du}}\definecolor{dialinecolor}{rgb}{0.074510, 0.023529, 0.843137}
\pgfsetstrokecolor{dialinecolor}
\draw (12.000000\du,26.000000\du)--(15.259361\du,17.000000\du)--(18.259361\du,19.000000\du)--(21.259361\du,15.000000\du)--(24.129681\du,27.000000\du)--(27.259381\du,13.000000\du)--(30.259381\du,17.000000\du)--(33.129681\du,15.000000\du)--(37.129681\du,31.000000\du);
}}
\pgfsetlinewidth{0.200000\du}
\pgfsetdash{}{0pt}
\pgfsetdash{}{0pt}
\pgfsetmiterjoin
\pgfsetbuttcap
{
\definecolor{dialinecolor}{rgb}{0.074510, 0.023529, 0.843137}
\pgfsetfillcolor{dialinecolor}
{\pgfsetcornersarced{\pgfpoint{0.000000\du}{0.000000\du}}\definecolor{dialinecolor}{rgb}{0.074510, 0.023529, 0.843137}
\pgfsetstrokecolor{dialinecolor}
\draw (-0.132319\du,31.043800\du)--(3.000000\du,12.000000\du)--(6.000000\du,21.000000\du)--(7.000000\du,20.000000\du)--(12.000000\du,26.000000\du);
}}
\end{tikzpicture}}
  \caption{Transforming an arbitrary alternating sequence into a
    feasible solution.}\label{drawing_proof}
\end{figure}
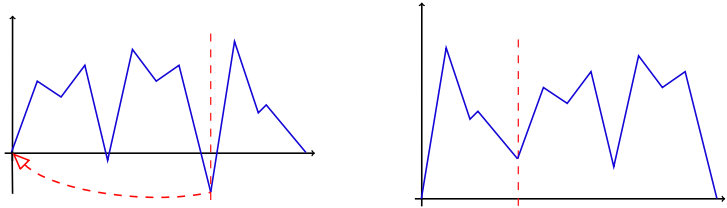
Furthermore, the upper bound of $2\mu$ is also tight for the
alternating problem.  If we modify the example given in
\cite{kellerer1998stock}, we have an example for the alternating
problem with an optimal stock size of $2p-3$, while $\mu=p$.
\begin{eqnarray*}
X  = \{\underbrace{p-1, \dots, p-1}_{p \text{ entries }}, 2\},\hspace{10mm}
Y  = \{\underbrace{p, \dots, p}_{p-1 \text{ entries }}, 1,1\}.\\
\end{eqnarray*}

In this section, we will present algorithms for the alternating stock
size problem.  We will use the notion of a $(q,T)$-pair, which is a
special case of a $(q,T)$-batch introduced and used by
\cite{kellerer1998stock} for the stock size problem.  
\begin{definition}{\cite{kellerer1998stock}}
A pair of jobs $\{x,y\}$, for $x\in X$ and $y \in Y$, is called
a $(q,T)$-pair for positive reals $T$ and $q \leq 1$, if:
\begin{enumerate}
\item[(i)] $x, y \leq T$, 
\item[(ii)] $|x-y| \leq qT$.
\end{enumerate}
\end{definition}
The following lemma is a special case of Lemma
3 in \cite{kellerer1998stock}, and the proofs are identical.
We provide the proof in Appendix \ref{appendix2} for the sake of
completeness.
\begin{lemma}\label{thm:batch_thm}
For positive $T$, $q \leq 1$ and a set of jobs partitioned into
$(q,T)$-pairs, we can find an alternating sequence of the jobs with
maximum stock size less than $(1+q)T$.
\end{lemma}

\subsection{The Pairing Algorithm} \label{sec:alg_pair}

We now consider the simple algorithm that pairs $x$- and $y$-jobs, and
then applies Lemma \ref{thm:batch_thm} to sequence the pairs.  Suppose
that there is some specific pairing that matches each $x_i$ to some
$y_j$, and consider the difference $x_i-y_j$ for each pair.  Let
$\alpha_1 \geq ... \geq \alpha_{n_1}$ denote the positive differences,
and let $\beta_1 \geq ... \geq \beta_{n_2}$ denote the absolute values
of the negative differences, where $n_1 + n_2 = n$.

\begin{lemma}\label{lem:alpha-beta}
The matching~$M^\star$ that matches~$x_i$ and~$y_i$ for
all~$i\in\{1,\ldots,n\}$ minimizes both $\alpha_1$ and~$\beta_1$.
\end{lemma}

\begin{proof}
Let~$M$ be an arbitrary matching that is different
from~$M^\star$. Then there exist edges~$(i_1,j_1)\in M$
and~$(i_2,j_2)\in M$ with~$i_1>i_2$ and~$j_1<j_2$. We show that we can
replace these edges by the edges~$(i_1,j_2)$ and~$(i_2,j_1)$ without
increasing~$\alpha_1$ or~$\beta_1$. From this the theorem follows
because after a finite number of such exchanges we obtain the
matching~$M^\star$.

Since $x_1\ge\ldots\ge x_n$ and~$y_1\ge\ldots\ge y_n$, we have $x_{i_1}\le x_{i_2}$ and $y_{j_1}\ge y_{j_2}$. This implies~$x_{i_1}-y_{j_1} \le x_{i_2}-y_{j_2}$ and 
\[
   \max(x_{i_1}-y_{j_1},x_{i_2}-y_{j_2})
   = x_{i_2}-y_{j_2} \ge \max(x_{i_1}-y_{j_2},x_{i_2}-y_{j_1}).
\]
The aforementioned inequalities also imply that
\[
   \min(x_{i_1}-y_{j_1},x_{i_2}-y_{j_2})
   = x_{i_1}-y_{j_1} \le \min(x_{i_1}-y_{j_2},x_{i_2}-y_{j_1}).
\]
Hence, neither~$\alpha_1$ nor~$\beta_1$ can increase due to the exchange.
\end{proof}

The pairing given by $M^\star$ directly results in a $2$-approximation
for the alternating stock size problem, by applying Lemma
\ref{thm:batch_thm}.  Without loss of generality, let us assume that
$\max\{\alpha_1, \beta_1\} = \alpha_1$, and observe
that~$\alpha_1\le\mu$.  Then $M^\star$ partitions the input into
$(\alpha_1/\mu,\mu)$-pairs.  Applying Lemma \ref{thm:batch_thm}, we
obtain an algorithm that computes a solution with value at most $(1+\alpha_1/\mu)\mu=\mu +
\alpha_1 \leq 2\mu$.  We note that if $\alpha_1 \leq (1-\epsilon)\mu$,
then we have a $(2-\epsilon)$-approximation.


\subsection{Lower Bound for the Alternating Stock Size
  Problem}\label{sec:lower_bound}

In order to obtain an approximation ratio better than $2$, we need to
use a lower bound that is more accurate than $\mu$.  We now introduce
a lower bound closely related to the one given for the stock size
problem by Kellerer et al. (Lemma 8 in \cite{kellerer1998stock}).  We
refer to a real number $C$, which divides the sets $X$ and $Y$ into
sets of small jobs and big jobs, as a \textit{barrier}.  Let $C \leq
\mu $ be a barrier such that
\begin{eqnarray}
X =  \{a_1 \geq a_2 \geq ... \geq a_{n_a} \geq & C & > v_k \geq v_{k-1}
\geq ... \geq v_{1}\}, \label{barrier1}\\ 
Y = \{~b_1 \geq b_2 \geq ... \geq b_{n_b} \geq & C & >
w'_1 \geq w'_2 \geq ...  \geq w'_{n_a-n_b} \geq w_1 \geq w_2 \geq ... \geq
w_{k}\}, \label{barrier2}
\end{eqnarray}
where, without loss of generality, $n_a \geq n_b$.  (If not, then by
swapping the $x$'s and the $y$'s we have a reverse (but equivalent) sequencing
problem with $n_a \geq n_b$).  The elements of \eqref{barrier1} are
all of the $x$-jobs (partitioned into the sets $A$ and $V$) and the
elements of \eqref{barrier2} are all of the $y$-jobs.  The jobs in $Y$
that have value less than $C$ are partitioned into $W'$ and $W$.

Let $A' =\{a_{n_b+1}, \dots, a_{n_a}\} = \{a'_1, \dots,
a'_{n_a-n_b}\}$, let $V_i$ denote the $i$ smallest $v_j$'s,
i.e.\ $\{v_1, v_2, \dots, v_i\}$, and let $W_i$ denote the $i$ largest
$w_j$'s in $W$, i.e.\ $\{w_1, w_2, \dots, w_i\}$.  (Note that $A'$,
$V_i,$ and $W_i$ each depend on $C$, but in order to avoid cumbersome
notation, we do not use superscript $C$.)  Let $s \in \{1, \dots,
n_a-n_b\}$.  After fixing a barrier $C$, let $h$ be the (unique) index
such that $w_h > v_h$ and $w_{h+1} \leq v_{h+1}$.  If no such index
$h$ exists because $w_1 < v_1$, then let $h=0$.  If no such index $h$
exists because $w_k > v_k$, then let $h=k$.  Recall that $S^*$ is the
value of an optimal ordering.  When $n_a > n_b$, then for any $s \in
\{1, \dots, n_a - n_b\}$, the following inequality applies.  For a
particular $s$, if the right-hand side of the inequality is positive,
then it yields a lower bound on $S^*$.

\begin{lemma}\label{lem:lower_bound}
Suppose $n_a > n_b$.  Let $s \in \{1, \dots, n_a-n_b\}$.
Then the following inequality holds:
\begin{eqnarray*}
S^* \geq LB(C) = \frac{1}{n_a-n_b-s+1} \cdot \left( 2 \sum_{i=s}^{n_a-n_b}{a'_i} -
\sum_{i=s}^{n_a-n_b}{w'_i} + \sum_{i=1}^h (v_i - w_i)\right).
\end{eqnarray*}
\end{lemma}

\begin{proof}
Following the proof of Lemma 8 in \cite{kellerer1998stock}, consider
the job sequence $L^*$ that is the optimal ordering restricted only to
the jobs with value at least $C$.  Then there are at least $n_a-n_b$
jobs in $A$ whose direct successor in $L^*$ is another job in $A$.
(If the very last job in $L^*$ is in $A$, then we say that the first
job in $L^*$ is its direct successor.)  Consider such a job $a_i$ and
its direct successor in $L^*$, $a_j$. Such a pair must be separated in
the optimal schedule by either a single job from $W'\cup W$ or by an
alternating sequence of jobs from $W'\cup W$ and $V$.  We refer to the
spaces in the optimal solution between~$a_i$ and~$a_j$ as {\em
  slots}. We will refer to the value~$a_i+a_j$ plus the total value of
the jobs in the corresponding slot as the value of the pair~$a_i$ and
$a_j$. Note that the value of the pair~$a_i$ and~$a_j$ is a lower
bound on $S^*$ for any pair~$a_i$ and~$a_j$ that is consecutive
in~$L^*$.

Consider the pairs of successive $a$'s in $L^*$ whose corresponding
slots do not contain jobs from the set $\{w_1', \dots, w_{s-1}'\}$.
There are at least~$(n_a-n_b-s+1)$ such pairs and we now consider
the~$(n_a-n_b-s+1)$ such pairs with the smallest values.  We will
determine a lower bound on the average value of these pairs. Being
pessimistic (we want to obtain a high lower bound, and this assumption
may make it lower), we assume that these pairs involve
the~$(n_a-n_b-s+1)$ smallest values in $A'$. Furthermore, we assume
that each of these values appears in two of the considered pairs.
Moreover, in order to decrease the lower bound even more, it may be
the case that all of the pairs in the set $\{(v_i,w_i)\mid 1\leq i
\leq h\}$ are placed in some slots.  Thus, the average value of the
considered pairs is at least
\[
   \frac{1}{n_a-n_b-s+1} \cdot \left(2\sum_{i=s}^{n_a-n_b}{a'_i} -
  \sum_{i=s}^{n_a-n_b}{w'_i} + \sum_{i=1}^h (v_i - w_i)\right).
\]
This directly leads to our lower bound, because there exists at least one pair whose value is at least the average value.
\end{proof}

\subsection{Alternating Batches: Definition}\label{sec:alt_batch}

We need a few more tools before we can outline our new algorithm.  The
notion of {\em batches} introduced in \cite{kellerer1998stock}, to which we
briefly alluded before Lemma \ref{thm:batch_thm}, is quite useful for
the stock size problem.  For $B \subseteq X \cup Y$, let $x(B)$ and
$y(B)$ denote the total value of the $x$-jobs and $y$-jobs,
respectively, in $B$.  In its original form, the batching lemma (Lemma
3, \cite{kellerer1998stock}) calls for a partition of the input into
groups or batches such that for some fixed positive real numbers $T$
and $q \leq 1$, each group $B$ has the following properties: $x(B),
y(B) \leq T$ and $|x(B)-y(B)| \leq qT$.  Given such a partition of the
input, a sequence with stock size at most $(1+q)T$ can be produced.

This approach is not directly applicable to the alternating stock size
problem, because the output is not necessarily an alternating
sequence.  However, we will now show that the procedure can be
modified to yield a valid ordering.  With this goal in mind, we define
a new type of batch, which we call an \textit{alternating batch}.  An
alternating batch will either contain two jobs (small) or more than
two jobs (large).

The modified procedure to construct an ordering of the jobs first
partitions the input into alternating batches, then orders these
batches, and finally orders the jobs contained within each batch.
In the case of a small alternating batch, the batch will contain both
an $x$-job and a $y$-job, and the last step simply preserves this
order.  A large alternating batch will be required to fulfill certain
additional properties that allow the elements to be sequenced in a way
that is both alternating and feasible, i.e. all prefix sums are
nonnegative.

Suppose
$B = \{(x'_1, y'_1), (x'_2, y'_2), \dots , (x'_{\ell}, y'_{\ell})\}$,
and consider the following four properties:
\begin{itemize}
\item[(i)] $\sum_{i=1}^{\ell} x'_i - \sum_{i=1}^{\ell}  y'_i \geq 0$,
\item[(ii)] $x'_1 - y'_1 \geq 0$,
\item[(iii)] $x'_i - y'_i \leq 0$, for $2 \leq i \leq \ell$,
\item[(iv)] $y'_1 \geq y'_2 \geq ... \geq y'_{\ell}$.
\end{itemize}
  
\begin{lemma} \label{lemma: alt_batch}
If a batch $B$ satisfies properties (i), (ii), (iii) and (iv), then we
can sequence the elements in $B$ so that the items alternate, each
prefix is nonnegative, and the maximum
height (or prefix sum) of the sequence is $x'_1$.
\end{lemma}

\begin{proof}
Place the items in the order: $x'_1, y'_1, x'_2, y'_2, \dots
x'_{\ell}, y'_{\ell}$.  All prefix sums of this sequence are
nonnegative, because by (ii) and (iii) only the first pair may have a
positive sum, and by (i) the
sum of all the pairs is nonnegative.
Since $x'_2 \leq y'_2 \leq
y'_1~ \Rightarrow ~x'_2 \leq y'_1$, after placing $x'_2$, we are
strictly less than height $x'_1$.  Repeating the argument, i.e. $x_i'
\leq y_i' \leq y_{i-1}' \Rightarrow x_i' \leq y_{i-1}'$, shows that
all prefix sums have value less than $x'_1$.
\end{proof}

\begin{definition}\label{def:alter_batch}
We call a set $B$ a \emph{$(1-\eps)$-alternating batch} if $B =
\{(x'_1, y'_1), (x'_2, y'_2), \dots , (x'_{\ell}, y'_{\ell})\}$ such
that
\begin{itemize}
  \item [(1)] $|\sum_{i=1}^{\ell} x'_i - \sum_{i=1}^{\ell}  y'_i |
    \leq (1-\epsilon)\mu$,
  \item [(2)] if $\ell > 1$, then conditions $(i)$ to $(iv)$ hold.
\end{itemize} 
\end{definition}

\begin{definition}\label{def:large_batch}
We say that a $(1-\eps)$-alternating batch with more than two jobs is
a {\em large} alternating batch.  In other words, a large alternating
batch obeys conditions (1) and (2) in Definition \ref{def:alter_batch}.
A {\em small} alternating batch contains only two jobs and obeys
condition (1) in Definition \ref{def:alter_batch}.
\end{definition}

Note that, by definition, in a large alternating batch $B$, the sum
of the $x$-jobs in $B$ is at least the sum of the $y$-jobs in $B$.

\begin{lemma} \label{lemma : alt_batch_alg}
If the sets $X$ and $Y$ can be partitioned into large and small
$(1-\eps)$-alternating batches, then we can find an alternating
sequence with maximum stock size less than $(2-\epsilon)\mu$.
\end{lemma}

\begin{proof}
We will show that the proof of Lemma 3 in \cite{kellerer1998stock} can
be modified to prove our lemma.  
(This proof is almost identical to that in \cite{kellerer1998stock},
but since we need to make subtle changes, we include it in its
entirety here for the sake of completeness.)
Let us set $q=(1-\eps)$ and $T=\mu$.
The only difference will be that inside the large alternating batches,
we will not always sequence all of the $x$'s before all of the $y$'s,
but we instead use the algorithm for sequencing an alternating batch
that was given in Lemma \ref{lemma: alt_batch}.

We sort all of the $q$-alternating batches based on the value of
$x(B)-y(B)$ in nondecreasing order into a sequence ${\mathscr B}$.
Let us begin with the empty list $L^S$ and with the current stock size~$S$
set to zero.  We repeat the following step until ${\mathscr B}$ is
empty:

``Find the first batch $B$ in ${\mathscr B}$ such that $S+x(B)-y(B)
\geq 0$ and set $S := S + x(B) -y(B)$.  Append $B$ to $L^S$ and remove
it from ${\mathscr B}$.''  (That such a batch $B$ exists follows from the facts that
$S \geq 0$ and that total value of the $x$- and $y$-jobs is zero.)

Afterward, we sequence each large alternating batch
$B$ according to Lemma \ref{lemma: alt_batch}, and each small
alternating batch by simply placing the $x$-job before the $y$-job.

Since the sum of all $x(B)-y(B)$ is zero, and the stock size never
goes below zero, each time a batch with positive $x(B)-y(B)$ is
chosen, there exists at least one unsequenced batch with negative
$x(B)-y(B)$.  To prove the upper bound $(1+q)T$ on the maximum,
\cite{kellerer1998stock} introduce the notion of {\em breakpoints},
which fulfill the following two conditions: (a) at each breakpoint,
the current stock size $S$ is less than $qT$, and (b) between any two
consecutive breakpoints, $S$ remains below $(1+q)T$.  Obviously, if
the breakpoints cover the whole time period, this will prove the
lemma.

The first break point is at time zero; the other breakpoints are the
time points just before a batch $B$ with positive $x(B)-y(B)$ is
started.  The last breakpoint is defined to be just after the last
batch.  

The first breakpoint and the last one fulfill condition (a) by
definition.  If one of the other breakpoints would not fulfill
condition (a), then $S \geq qT$ must hold, and because of property (1)
in Definition \ref{def:alter_batch}, our algorithm would have chosen a
batch $B$ with negative $x(B)-y(B)$ as the next batch.  Thus, all of
the breakpoints fulfill the condition that $S < qT$.

Now we need to consider the values of $S$ between two consecutive
breakpoints.  Let us consider two consecutive breakpoints $BP_i$ and
$BP_{i+1}$.  
Recall that all batches $B^-$ with
nonpositive $x(B^-)-y(B^-)$ have only two jobs.  
Since, at time $BP_i$, a batch $B^+$ with positive
$x(B^+)-y(B^+)$ is started, it follows that for each batch $B^-$, the
inequality
\begin{eqnarray}
S_i + x(B^-) < T \label{less_than_T}
\end{eqnarray}
holds.  Otherwise, $S_i + x(B^-)-y(B^-) \geq 0$, because $y(B^-)$ is a
single job and is therefore at most~$T$.

After batch $B^+$ is appended to $L^S$, the current stock size
increases to $S_i + x(B^+)-y(B^+) \leq S_i + qT$.  If batch $B^+$
contained only two jobs, then in between the stock size is at most
$S_i + x(B^+) < qT+T$.  If $B^+$ is a large alternating batch, then by
Lemma \ref{lemma: alt_batch}, the highest point after $S_i$ is at most
$S_i + x_1' < qT + T$.  Either the next batch again has positive
$x(B)-y(B)$ (then we have made it to the next breakpoint) or there
follows a sequence of (small) batches with nonpositive $x(B)-y(B)$.
The stock size within any of these batches $B^-$ always remains below
\begin{eqnarray*}
S_i + x(B^+) - y(B^+) + x(B^-) = S_i + x(B^-) + \left(x(B^+)
-y(B^+)\right) < T + qT,
\end{eqnarray*}
because of inequality \eqref{less_than_T} and because $B^+$ is a
$q$-alternating batch.  After each of these batches, the stock size does
not increase.  This shows that condition (b) holds for any two
consecutive breakpoints, and the proof of the lemma is complete.
\end{proof}

\subsection{Alternating Batches: Construction}\label{sec:last_case}

In this section, we present the final tool required for our algorithm.
Suppose that for some some $\eps: 0 \leq \eps \leq 1$, the following
conditions hold for an input instance to the alternating stock size
problem:
\begin{itemize}
\item $\alpha_1 > (1-\eps)\mu$,
\item $LB(C) < \frac{2}{2-\eps}\mu,$ for $C=(1-\eps)\mu$. 
\end{itemize}
Then, we claim, there is some value of $\eps$ (to be determined later)
for which these two conditions can be used to partition the input
into $(1-\eps)$-alternating batches, to which we can then apply Lemma
\ref{lemma : alt_batch_alg}.  In this section, we will heavily rely on
the notation introduced in Section \ref{sec:lower_bound}.

The sets $A'=\{a_1', \dots, a_{n_a-n_b}'\}$ and $W'=\{w_1',\dots,
w_{n_a-n_b}'\}$ contain exactly the pairs in $M^\star$ that are split
by barrier $C$.  Let $s$ be the smallest index such that $w'_s <
\epsilon\mu$.  To see that such an $s$ actually exists, we note the
following.  Let $i^\star$ denote the index such that $x_{i^\star} -
y_{i^\star} = \alpha_1$.  Then $y_{i^\star} < \eps\mu$ and the
pair $(x_{i^\star}, y_{i^\star})$ is split by $C$.  Thus,
$y_{i^\star}$ corresponds to some $w_{i'}'$, and therefore $s \leq
i'$.  See Figure \ref{fig:batches_figure} for a schematic drawing.

For $i$ in $\{1,\dots,n_a-n_b\}$, we
define $\alpha'_i = a'_i - w'_i$ and for $j$ in $\{1,\dots,h\}$,
$\beta'_j = w_j-v_j$.  (Recall that for $j \in \{1,\dots,h\}$,
$w_j-v_j >0$.)  Furthermore, let $\mathcal{A}_i$ denote the pair
$\{a_i',w_i'\}$ and let $\mathcal{B}_j$ denote the pair $\{v_j,w_j\}$.
Since $w_s' < \eps\mu$, it follows that all $w_i$'s in $W$ also have
value less than $\eps\mu$.  Moreover, $\beta_j' < \eps \mu$ for $j
\in \{1, \dots, h\}$.

\begin{figure}\centering
\begin{tikzpicture}
\begin{scope}[shift={(0,0)}]
\draw (0,0) node {$a_1$};
\draw (0.6,0) node {$\cdots$};
\draw (1.2,0) node {$a_{n_b}$};

\draw (0,-2) node {$b_1$};
\draw (0.6,-2) node {$\cdots$};
\draw (1.2,-2) node {$b_{n_b}$};
\end{scope}

\begin{scope}[shift={(2.0,0)}]
\draw (0,0.8) node {$a_{n_b+1}$};
\draw (0,0.4) node {$=$};
\draw (0,0) node {$a'_{1}$};
\draw (0.6,0) node {$\cdots$};
\draw (1.3,0) node {$a'_{s}$};
\draw (2.0,0) node {$\cdots$};
\draw (2.6,0) node {$a'_{i'}$};
\draw (3.2,0) node {$\cdots$};
\draw (4.1,0.8) node {$a_{n_a}$};
\draw (4.1,0.4) node {$=$};
\draw (4.1,0) node {$a'_{n_a-n_b}$};

\draw (0,-2) node {$w'_{1}$};
\draw (0.6,-2) node {$\cdots$};
\draw (1.3,-2) node {$w'_{s}$};
\draw (2.0,-2) node {$\cdots$};
\draw (2.6,-2) node {$w'_{i'}$};
\draw (3.2,-2) node {$\cdots$};
\draw (4.1,-2) node {$w'_{n_a-n_b}$};
\end{scope}

\begin{scope}[shift={(7.2,0)}]
\draw (0,0) node {$v_k$};
\draw (0.6,0) node {$\cdots$};
\draw (1.3,0) node {$v_{h+1}$};
\draw (2.2,0) node {$v_{h}$};
\draw (3.0,0) node {$\cdots$};
\draw (4,0) node {$v_1$};

\draw (0,-2) node {$w_1$};
\draw (1.0,-2) node {$\cdots$};
\draw (1.7,-2) node {$w_h$};
\draw (2.6,-2) node {$w_{h+1}$};
\draw (3.4,-2) node {$\cdots$};
\draw (4,-2) node {$w_k$};
\end{scope}

\draw [decorate,decoration={brace,amplitude=10pt}]
(-0.2,1.2) -- (6.6,1.2) node [black,midway,xshift=-0.1cm,yshift=0.6cm] {\footnotesize$\ge (1-\epsilon)\mu$};

\draw [decorate,decoration={brace,amplitude=10pt}]
(1.5,-2.5) -- (-0.2,-2.5) node [black,midway,xshift=-0.1cm,yshift=-0.7cm] {\footnotesize$\ge (1-\epsilon)\mu$};

\draw [decorate,decoration={brace,amplitude=10pt}]
(11.5,-2.5) -- (3.0,-2.5) node [black,midway,xshift=-0.1cm,yshift=-0.7cm] {\footnotesize$< \epsilon \mu$};

\draw (2.0,-0.4) -- (2.0,-1.6);
\draw (2.3,-1) node {\footnotesize$\alpha_1'$};
\draw (3.3,-0.4) -- (3.3,-1.6);
\draw (3.6,-1) node {\footnotesize$\alpha_s'$};
\draw (4.6,-0.4) -- (4.6,-1.6);
\draw (5.1,-0.8) node {\footnotesize$\alpha_{i'}'$};
\draw (5.1,-1.2) node {\footnotesize$=\alpha_1$};
\draw (6.1,-0.4) -- (6.1,-1.6);
\draw (6.7,-1) node {\footnotesize$\alpha_{n_a-n_b}'$};

\draw (7.2,-1.6) -- (11.2,-0.4);
\draw (8.3,-1) node {\footnotesize$\beta_1'$};
\draw (8.9,-1.6) -- (9.4,-0.4);
\draw (9.0,-0.6) node {\footnotesize$\beta_h'$};
\end{tikzpicture}
\caption{An illustration of the various elements used in the
  construction of the lower bound.\label{fig:batches_figure}}
\end{figure}
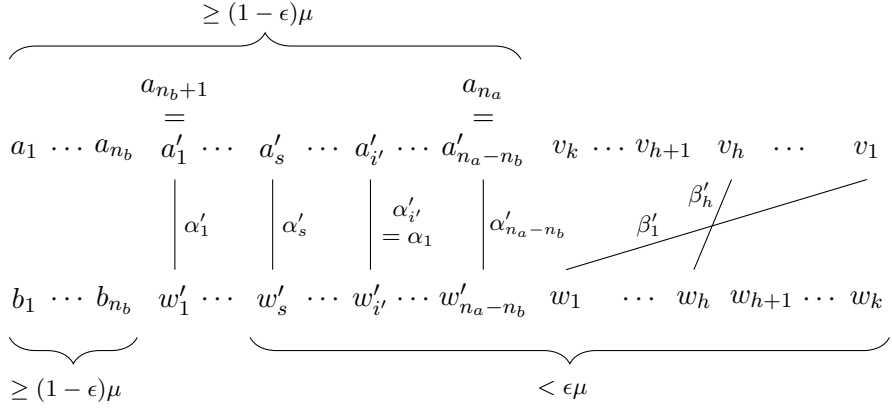

Our goal is now to construct $(1-\eps)$-alternating batches.  For each
$i \in \{1, \dots, s-1\}$, note that $\alpha_i' \leq (1-\eps)\mu$.  The
set $\mathcal{A}_i$ therefore forms a small $(1-\eps)$-alternating
batch.  For each $\mathcal{A}_i$ where $i \in \{s, \dots, n_a-n_b\}$,
we will find a set of $\mathcal{B}_j$'s that can be grouped with this
$\mathcal{A}_i$ to create a large $(1-\eps)$-alternating batch.
However, to do this, we require that the condition on $\eps$ found in
Claim~\ref{claim:eps} be satisfied.
\begin{claim}\label{claim:eps}
For $\eps \leq .219$, the following inequality is satisfied.
\begin{eqnarray*}
2(1-\epsilon) - \frac{2}{2-\epsilon} > 2\epsilon.
\end{eqnarray*}
\end{claim}

\begin{lemma}\label{lem:condition}
If $LB(C) < 2\mu/(2-\eps)$, $C = (1-\eps)\mu$, and
$2(1-\epsilon) - \frac{2}{2-\epsilon} > 2\epsilon$, then
  $\sum_{i=1}^h \beta'_i + \sum_{i=s}^{n_a-n_b}w'_i > 2\epsilon\mu(n_a-n_b-s+1)$.
\end{lemma}

\begin{proof}
By the first assumption in the statement of the Lemma, we have
\begin{eqnarray*}
LB(C) = \left( 2 \sum_{i=s}^{n_a-n_b}{a'_i} - \sum_{i=s}^{n_a-n_b}{w'_i} +
\sum_{i=1}^h (v_i - w_i)\right) \cdot \frac{1}{n_a-n_b-s+1} <
\frac{2\mu}{2-\epsilon}.
\end{eqnarray*}
Since
$C = (1-\epsilon)\mu$ and each $a_i' \geq C$, we also have:
$$\sum_{i=s}^{n_a-n_b}{a'_i} \geq (n_a-n_b-s+1)(1-\epsilon)\mu.$$
Rearranging and multiplying each side by 2, we have
\begin{eqnarray*}
\frac{2 \sum_{i=s}^{n_a-n_b}{a'_i}}{(n_a-n_b-s+1)} \geq
2(1-\epsilon)\mu.
\end{eqnarray*}
Therefore,
\begin{eqnarray*}
2(1-\epsilon)\mu - \left( \sum_{i=s}^{n_a-n_b}{w'_i} + \sum_{i=1}^h
\beta'_i\right) \cdot \frac{1}{n_a-n_b-s+1} \leq LB(C) < 2\mu/(2-\epsilon).
\end{eqnarray*}
By the condition on $\epsilon$, we have
\begin{eqnarray*}
2 \eps \mu < \left(2(1-\epsilon) - \frac{2}{2-\epsilon} \right) \mu
< \left( \sum_{i=s}^{n_a-n_b}{w'_i} + \sum_{i=1}^h
\beta'_i\right) \cdot \frac{1}{n_a-n_b-s+1}.
\end{eqnarray*}
We can conclude that
$\sum_{i=1}^h \beta'_i + \sum_{i=s}^{n_a-n_b}w'_i >
2\epsilon\mu(n_a-n_b-s+1)$.
\end{proof}

For ease of notation, we set $d=n_a-n_b-s+1$.  In the following lemma,
we show that we can also construct a $(1-\eps)$-alternating batch for
each $\mathcal{A}_i$ for $i \in [s,n_a-n_b]$.

\begin{lemma}\label{lem:special_batch}
For~$\eps = .21$, there exists $d$ disjoint subsets $S_1, \dots, S_d$
of $\{\mathcal{B}_1,\dots,\mathcal{B}_h\}$ such that for all $i$ in
$\{1,\dots,d\}$, the set $S_i \cup \mathcal{A}_{i+s-1}$ is a
$(1-\eps)$-alternating batch.
\end{lemma}

\begin{proof}
Our goal is to show that a set of $\mathcal{B}_j$'s can
be assigned to each $\mathcal{A}_i$ so that the total value of the
corresponding set of elements is at most $(1-\eps)\mu$.  
This will 
imply that condition (1)
in Definition~\ref{def:alter_batch} holds.  Note that conditions (ii),
(iii) and (iv) hold for any set $S_i \cup \mathcal{A}_{i+s-1}$.  We will also
show that (i) holds for the batches we construct.

For a subset
$S$ of $\{\mathcal{B}_1,\dots,\mathcal{B}_h\}$, let $f(S)$ denote
the sum of the weight of the elements in $S$.  We will show
that for each $i \in \{1, \dots, d\}$, we can find a disjoint set $S_i$ such
that $f(S_i)$ lies in the interval $[\eps \mu - w_{i+s-1}, ~
2 \eps \mu - w_{i+s-1}]$.  This will imply that the set 
$S_i \cup \mathcal{A}_{i+s-1}$ has value in the interval $[0,
  (1-\eps)\mu]$ and is therefore a $(1-\eps)$-alternating batch.
Our algorithm is simply to form a set of the next
available (i.e. unused)
$\mathcal{B}_j$'s until their sum lies in the desired interval.

For $p$ in $\{1,\dots,d\}$, let $B_p = \{\mathcal{B}_1,\dots,\mathcal{B}_h\}
\setminus \cup_{k=1}^{p-1}S_k$ and recall that $f(B_p)$ denotes the sum of the
weight of the elements that are in $B_p$.  As we construct the sets
$S_i$, we will show at each step $p$ that the following hypothesis
holds: $f(B_p) + \sum_{t=s+p-1}^{n_a-n_b}w'_t >
2\epsilon\mu(d-(p-1))$, and we have formed $p-1$ sets $S_1, \dots,
S_{p-1}$ such that each $S_k \cup A_{k+s-1}$ is a
$(1-\eps)$-alternating batch for $k \in \{1, \dots, p-1\}$.

When $p=1$, the hypothesis is given by Lemma \ref{lem:condition}.
Now, let's assume that we have made $p-1$ sets for $p-1<d$. If
$\alpha'_{p+s-1} \leq (1-\epsilon)\mu$, we set $S_{p} = \emptyset$ and the
required inequality still holds.  Otherwise, let $S_{p}$ be a subset
of $B_{p}$ such that $\epsilon\mu - w'_{s+p-1} \leq f(S_{p}) <
2\epsilon\mu - w'_{s+p-1}$.  Such a subset exists because all the
elements of $B_{p}$ are at most $\epsilon\mu$, and because 
$f(B_p)+w'_{s+p-1} \geq \epsilon\mu$, which follows from the induction hypothesis.
The induction hypothesis implies that $f(B_p) +
\sum_{t=s+p-1}^{n_a-n_b}w'_t > 2\epsilon\mu(d-(p-1))$. Furthermore,
$w'_t\le \epsilon\mu$ for all~$t$. Therefore,
\begin{align*}
   f(B_p) + w'_{s+p-1} & > 2\epsilon\mu(d-(p-1)) - \sum_{t=s+p}^{n_a-n_b}w'_t\\
   & \ge 2\epsilon\mu(d-(p-1)) - \sum_{t=s+p}^{n_a-n_b}\epsilon\mu\\
   & = 2\epsilon\mu(d-(p-1)) - (n_a-n_b-s-p+1)\epsilon\mu\\
   & = 2\epsilon\mu(d-(p-1)) - (d-p)\epsilon\mu\\
   & = \epsilon\mu(d-p+2) \ge \epsilon\mu.
\end{align*}

Then $f(B_{p+1}) > f(B_p) - 2\epsilon\mu + w'_{s+p-1} >
2\epsilon\mu(d-(p-1)) - \sum_{t=s+p-1}^{n_a-n_b}w'_t - 2\epsilon\mu
+ w'_{s+p-1}$, and $f(B_{p+1}) + \sum_{t=s+p}^{n_a-n_b}w'_t >
2\epsilon\mu(d-p)$. So the inequality holds at step $p+1$ and we
have constructed the set $S_p$.  This proves the lemma, since for
every~$p$
  \[
     a'_{p+s-1}-(w'_{s+p-1}+f(S_p)) \in [0,(1-\epsilon)\mu],
  \]
  which follows from~$a'_{p+s-1}\in[(1-\epsilon)\mu,\mu]$ and~$w'_{s+p-1}+f(S_p)\in[\epsilon\mu,2\epsilon\mu]$.
\end{proof} 

Now we want to complete the construction of the $(1-\eps)$-alternating
batches, so that we can apply Lemma \ref{lemma : alt_batch_alg}.
For the sets $\mathcal{A}_i$, where $i \in \{s, \dots, n_a-n_b\}$, we
construct batches according to Lemma \ref{lem:special_batch}.  Let
$y_{i^*} = w_{s}'$.  For all $i < i^*$, the pair $(x_i,y_i)$
forms a small $(1-\eps)$-alternating batch.  This follows from the fact
that for all $i < i^*$, $y_{i} \geq \eps \mu$, by definition of $s$.
Finally, if there are remaining elements, they are $v_i$'s and
$w_i$'s, which can be paired up arbitrarily to construct more small
$(1-\eps)$-alternating batches, since each remaining
$v_i$ has value strictly less than $(1-\epsilon)\mu$ due to our choice of
barrier, and each remaining $w_i$ has value at most $\epsilon\mu$.

Since the only limits on the value of $\eps$ are imposed by Lemma 
\ref{lem:condition}, we can set $\eps=.21$ and partition the input
into .79-alternating batches.

\subsection{A {1.79}-Approximation Algorithm}

We are now ready to present an algorithm for the alternating stock
size problem with an approximation guarantee of $1.79$.
  \begin{algorithm}[H]
    \caption{$1.79$-approximation}
    \begin{algorithmic}[1]
    \State {\bf Input:} the sets $X$ and $Y$ of positive numbers sorted in nonincreasing order.
    \State {\bf Output:} a sequence that is a $1.79$-approximation.
    \State Set $\eps=.21, C = (1-\eps)\mu$.
    \State Match each $x_i$ with $y_i$.
    \If {$\alpha_1 \leq (1-\epsilon)\mu$ or if $LB(C) \geq
    \frac{2}{2-\eps}\mu$}
        \State \Return solution for the Pairing Algorithm with
        guarantee of most $\mu + \alpha_1$.
        \Else 
            \State Partition the input into $(1-\eps)$-alternating
            batches as described in Section \ref{sec:last_case}.
            \State Run the algorithm from Lemma \ref{lemma : alt_batch_alg} on the
$(1-\eps)$-alternating batches.
            \EndIf 
  \end{algorithmic}
  \end{algorithm}

\begin{theorem}\label{thm:alternating}
Algorithm 1 is a $1.79$-approximation for the alternating stock size problem. 
\end{theorem}
  
\begin{proof}
In the first case, we have $\alpha_1 \leq (1-\epsilon)\mu$. The
algorithm described in Section \ref{sec:alg_pair} therefore gives a
solution whose value is at at most $\mu + \alpha_1 \leq
(2-\epsilon)\mu$, and we know that $\mu$ is a lower bound. 
In the second case, we have $LB(C) \geq 2\mu/(2-\eps)$, in which case an
algorithm with a guarantee of $2\mu$ is a $(2-\eps)$-approximation.
The last case is covered in the proof of Lemma \ref{lemma :
  alt_batch_alg}.
\end{proof}




\section{Gasoline Problem}
\label{section:GasolineProblem}

Let the variable $z_{ij}$ be $1$ if gas station $x_i$ is placed in position
$j$, and be $0$ otherwise.  Then we can formulate the gasoline problem as the following integer linear
program whose solution matrix $Z$ is a permutation matrix.
\begin{align}
 &\min \beta-\alpha \notag\\
&\forall j \in [1,n]: \sum_{i=1}^n z_{ij}  = 1, \quad
\forall i \in [1,n]: \sum_{j=1}^n z_{ij} = 1, \quad \forall i,j \in
        [1,n]: z_{ij}  \in \{0,1\},\notag\\
&\forall k\in\{1,\ldots,n\} : \sum_{j=1}^{k} \sum_{i=1}^n z_{ij}
        \cdot x_i - \sum_{j=1}^{k-1}  y_j  \leq
        \beta, \label{small-objective}\\
&\forall k\in\{1,\ldots,n\} : \sum_{j=1}^{k} \sum_{i=1}^n z_{ij}
        \cdot x_i - \sum_{j=1}^{k}  y_j  \geq
        \alpha. \label{small-objective2}
\end{align}
Observe that~\eqref{small-objective} and~\eqref{small-objective2}
imply that for every interval~$I=[k,\ell]$ the sum of the $x_i$'s
assigned to~$I$ by $Z$ and the sum of the $y_i$'s in~$I$ differ by at
most~$\beta-\alpha$.  If we replace~$z_{ij} \in \{0,1\}$ with the
constraint $z_{ij}\in[0,1]$, then the solution to the linear program,
$Z$, is an $n \times n$ doubly stochastic matrix.  Now we have the
following rounding problem. We are given an $n \times n$ doubly
stochastic matrix~$Z = \{z_{ij}\}$
and we define $z_j$ to be the total fractional value of the $x_i$'s that
are in position~$j$, i.e. $z_j = \sum_{i=1}^n z_{ij}\cdot x_i$.
Our goal is to find a permutation of the $x_i$'s such that the $x_i$
assigned to position~$j$ is roughly equal to $z_j$.

A natural approach would be to decompose $Z$ into a convex combination
of permutation matrices and see if one of these gives a good
permutation of the elements in $X$.  However, consider the following
example:
\begin{eqnarray*}
X & = & \{\underbrace{1,1, \dots, 1}_{n-k \text{ entries
}},\underbrace{B, B, \dots, B}_{k \text{ entries
}}\},\quad \forall i \in [1,n]:~y_i = \gamma = \frac{k\cdot B + n-k}{n}.
\end{eqnarray*}
In this case, $z_j = \gamma$ for all $j \in [1,n]$.  Thus, a possible
decomposition into permutation matrices could look like:
\begin{eqnarray*}
&\{B,B, \dots, B, 1, 1, \dots, 1, 1\}\\
&\{1, B,B, \dots, B, 1,  1,\dots, 1\}\\
&\vdots \\
&\{1, 1, \dots, 1, 1, B, B, \dots, B\}.
\end{eqnarray*}
Each of these permutations has an interval with very large
value, while the optimal permutation of the elements in $X$ is
\begin{eqnarray*}
\{1, 1, \dots 1, B, 1 \dots, 1, B, 1,\dots 1\}.
\end{eqnarray*}

\subsection{Transformation}

Given a doubly stochastic matrix~$Z = \{z_{ij}\}$, we transform it
into a doubly stochastic matrix~$T = \{t_{ij}\}$ with special
properties. First of all, for each~$j$, $z_j = \sum_{i=1}^{n} t_{ij}
\cdot x_i$.  This means that if~$(Z,\alpha,\beta)$ is a feasible
solution to the linear program then~$(T,\alpha,\beta)$ is also a
feasible solution. In particular, if~$Z$ is an optimal solution, for
which~$\beta-\alpha$ is as small as possible, then~$T$ is also
optimal.

We call a row $i$ in a doubly stochastic matrix $A=\{a_{ij}\}$ {\em
  finished} at column~$\ell$ if~$\sum_{j=1}^{\ell} a_{ij} = 1$. 
We say that a matrix $T$ has the {\em consecutiveness property}
if the following holds:
for each column~$j$ and any rows~$i_1$
and~$i_3$ with~$i_1<i_3$, $t_{i_1j}>0$, and~$t_{i_3j}>0$, each
row~$i_2\in\{i_1+1,\ldots,i_3-1\}$ is finished at column~$j$. 

Our procedure to transform the matrix~$Z$ into a matrix~$T$ with the
desired property relies on the following transformation rule. Assume
that there exist indices~$j$, $i_1$, $i_3$,
and~$i_2\in\{i_1+1,\ldots,i_3-1\}$ such that~$z_{i_1j}>0$,
$z_{i_3j}>0$, and row~$i_2$ is not finished in matrix~$Z$ at
column~$j$. Then the procedure {\sc shift} shown as Algorithm~\ref{alg:Shift} computes a column
vector~$a=(a_1,\ldots,a_n)$, which satisfies the following lemma.

\begin{algorithm}
  \caption{{\sc shift}($Z,j,i_1,i_2,i_3,\delta$)}\label{alg:Shift}
  \begin{algorithmic}[1]
    \State $\forall i\in\{1,\ldots,n\}\setminus\{i_1,i_2,i_3\}: a_{i} = z_{ij}$;
    \State $a_{i_2} = z_{i_2j}+\delta$;
      \If {$x_{i_1}=x_{i_3}$}
         \State $a_{i_1}=z_{i_1j}-\delta$; \quad $a_{i_3}=z_{i_3j}$;
      \Else
         \State $a_{i_1}=z_{i_1j}-\delta\cdot\frac{x_{i_2}-x_{i_3}}{x_{i_1}-x_{i_3}}$; \quad
          $a_{i_3}=z_{i_3j}-\delta\cdot\frac{x_{i_1}-x_{i_2}}{x_{i_1}-x_{i_3}}$;
      \EndIf
    \State\Return $a$
  \end{algorithmic}
\end{algorithm}

\begin{lemma}\label{lemma:TransformDoesNotChangeValue}
For any~$\delta\ge 0$, the vector~$a$ returned by {\sc shift}($Z,j,i_1,i_2,i_3,\delta$) satisfies
$\sum_{i=1}^{n} a_{i} \cdot x_i = z_j$.
\end{lemma}
\begin{proof}
Due to~$a_{i} = z_{ij}$ for all~$i\in\{1,\ldots,n\}\setminus\{i_1,i_2,i_3\}$, it suffices to prove that
\[
   a_{i_1}x_{i_1} + a_{i_2}x_{i_2} + a_{i_3}x_{i_3} = z_{i_1j}x_{i_1} + z_{i_2j}x_{i_2} + z_{i_3j}x_{i_3}.
\]
In the first case~$x_{i_1}=x_{i_3}$, this follows easily because in this case~$x_{i_1}=x_{i_2}=x_{i_3}$
(remember that~$i_1<i_2<i_3$, which implies~$x_{i_1}\ge x_{i_2} \ge x_{i_3}$).
In the second case~$x_{i_1}>x_{i_3}$, we have
\begin{align*}
   & z_{i_1j}x_{i_1} + z_{i_2j}x_{i_2} + z_{i_3j}x_{i_3} - (a_{i_1}x_{i_1} + a_{i_2}x_{i_2} + a_{i_3}x_{i_3})\\
   & = \delta\cdot\frac{x_{i_2}-x_{i_3}}{x_{i_1}-x_{i_3}}\cdot x_{i_1} - \delta\cdot x_{i_2}
       +\delta\cdot\frac{x_{i_1}-x_{i_2}}{x_{i_1}-x_{i_3}}\cdot x_{i_3}\\
   & = \frac{\delta}{x_{i_1}-x_{i_3}} \cdot\Big( (x_{i_2}-x_{i_3})x_{i_1} - (x_{i_1}-x_{i_3})x_{i_2} +(x_{i_1}-x_{i_2})x_{i_3}\Big) = 0.\qedhere
\end{align*}
\end{proof}

Let~$Z'$ denote the matrix that we obtain from~$Z$ if we replace the $j^{\text{th}}$ column by the vector~$a$ returned by the procedure {\sc shift}. The previous lemma shows that~$Z'$ satisfies~\eqref{small-objective} and~\eqref{small-objective2} for the same~$\beta$ and~$\alpha$ as~$Z$ because the value~$z_j$ is not changed by the procedure. However, the matrix~$Z'$ is not doubly stochastic because the rows~$i_1$, $i_2$, and~$i_3$ do not add up to 1 anymore. In order to repair this, we have to apply the {\sc shift} operation again to another column with~$-\delta$. Formally, let us redefine the matrix~$Z'= \{z_{ij}'\}$ as the outcome of the operation {\sc transform} shown as Algorithm~\ref{alg:Transform}.
\begin{algorithm}
  \caption{{\sc transform}($Z,j,i_1,i_2,i_3$)}\label{alg:Transform}
  \begin{algorithmic}[1]
    \State The $j^{\text{th}}$ column of~$Z'$ equals~{\sc shift}($Z,j,i_1,i_2,i_3,\delta$) for~$\delta>0$ to be chosen later.
  \State Let~$j'>j$ denote the smallest index larger than~$j$
  with~$z_{i_2j'}>0$. Such an index must exist because row~$i_2$ is
  not finished in~$Z$ at column~$j$. The $(j')^{\text{th}}$ column of~$Z'$ equals~{\sc shift}($Z,j',i_1,i_2,i_3,-\delta$).  
  \State All columns of~$Z$ and~$Z'$, except for columns~$j$ and~$j'$,
  remain unchanged.
  \State The value~$\delta$ is chosen as the largest value for which all entries of~$Z'$ are in~$[0,1]$. This value must be strictly larger than~$0$ due to our choice of~$j$, $j'$, $i_1$, $i_2$, and~$i_3$.
    \State\Return $Z'$
  \end{algorithmic}
\end{algorithm}

Observe that~$Z'$ is a doubly stochastic matrix because the rows~$i_1$, $i_2$, and~$i_3$ sum up to 1 and all entries are from~$[0,1]$. Applying Lemma~\ref{lemma:TransformDoesNotChangeValue} twice implies that~$(Z',\beta,\alpha)$ is a feasible solution to the linear program if~$(Z,\beta,\alpha)$ is one. 

We will transform~$Z$ by a finite number of applications of the operation {\sc transform}. As long as the current matrix~$T$ (which is initially chosen as~$Z$) does not have the consecutiveness property, let~$j$ be the smallest index for which there exist indices~$i_1$, $i_3$, and~$i_2\in\{i_1+1,\ldots,i_3-1\}$ such that~$t_{i_1j}>0$, $t_{i_3j}>0$, and row~$i_2$ is not finished in~$T$ at column~$j$. Furthermore, let~$i_1$ and~$i_3$ be the smallest and largest index with~$t_{i_1j}>0$ and~$t_{i_3j}>0$, respectively, and let~$i_2$ be the smallest index from $\{i_1+1,\ldots,i_3-1\}$ for which row~$i_2$ is not finished at column~$j$. We apply the operation~$\text{\sc transform}(T,j,i_1,i_2,i_3)$ to obtain a new matrix~$T$. 
\begin{lemma}\label{lemma:PolynomialNumberTransform}
After at most a polynomial number of {\sc transform} operations, no further such operation can be applied. Then~$T$ is a doubly stochastic matrix with the consecutiveness property.
\end{lemma}
\begin{proof}
If the {\sc transform} operation is not applicable anymore, then by definition the current matrix~$T$ must satisfy the consecutiveness property. Hence, we only need to show that this is the case after at most a polynomial number of {\sc transform} operations.

First of all observe that the smallest index~$j$ for which column~$j$ does not satisfy the consecutiveness property cannot decrease because {\sc transform} does not change the columns~$1,\ldots,j-1$. Hence, we only need to argue that~$j$ increases after a polynomial number of {\sc transform} operations. For this, observe that the smallest index~$i_1$ with~$t_{i_1j}>0$ cannot decrease and that the largest index~$i_3$ with~$t_{i_3j}>0$ cannot increase because the {\sc transform} operation only increases~$t_{i_2j}$ for some~$i_2$ with~$i_1<i_2<i_3$. Hence, again it is sufficient to prove that either~$i_1$ increases or~$i_3$ decreases after a polynomial number of steps. This follows from the fact that as long~$j$, $i_1$, and~$i_3$ do not change,~$i_2$ cannot decrease. Furthermore, as long as~$j$, $i_1$, $i_2$, and $i_3$ do not change, the index~$j'$ increases with every {\sc transform} operation. Hence, after at most~$n$ steps~$i_2$ has to increase, which implies that after at most~$n^2$ steps~$i_1$ has to increase or~$i_3$ has to decrease. 
\end{proof}

In the remainder, we will not need the matrix~$Z$ anymore but only matrix~$T$. For convenience, we will use the notation
$t_j = \sum_{i=1}^n t_{ij}\cdot x_i$ instead of~$z_j$ even though the
transformation ensures that~$t_j$ and~$z_j$ are equal.

We now define a graph whose connected components or {\em blocks}
will correspond to the row indices from columns that overlap.  More
formally, let $V=\{1,\ldots,n\}$ denote a set of vertices and
let $G_0$ be the empty graph on~$V$. Each column~$j$ of~$T$ defines a set~$E_j$ of edges as follows:
the set~$E_j$ is a clique on the vertices~$i\in V$ with~$t_{ij}>0$, i.e. $E_j$ contains an edge between two vertices~$i$ and~$i'$ if and only if~$t_{ij}>0$ and~$t_{i'j}>0$. We denote by~$G_j$ the graph on~$V$ with edge set~$E_1\cup\ldots\cup E_j$. 

\begin{definition}
A {\em block} in $G_j$ is a set of indices in
$[1,n]$ that forms a connected component in $G_j$.
A block in $G_j$ is called finished if all rows in $T$ corresponding to the
indices it contains are finished at column $j$.  Similarly, if a block
in $G_j$ contains at least one unfinished row at column $j$, it is
called an {\em unfinished block}.
\end{definition}

If~$B\subseteq\{1,\ldots,n\}$ is a block in~$G_j$ with~$i\in B$ then we will say that \emph{block~$B$ contains row~$i$}. For the following lemma, it is convenient to define a matrix~$C=\{c_{ij}\}$, which is the cumulative version of~$T$. To be more precise, the $j^{\text{th}}$ column of~$C$ equals the sum of the first~$j$ columns of~$T$.

\begin{lemma}\label{lem:BlockStructure}
The following three properties are satisfied for every~$j$.
\begin{enumerate}
\item Let $B$ be a block in $G_j$ and let~$k=\sum_{i\in B}c_{ij}$ the denote the value of block~$B$ at column~$j$. The number of rows in~$B$ is~$k$ if~$B$ is finished and it is~$k+1$ if~$B$ is an unfinished block. 
\item The set of blocks in~$G_j$ emerges from the set of blocks in~$G_{j-1}$ by either merging exactly two unfinished blocks or by making one unfinished block finished.
\item Let~$B_1,\ldots,B_{\ell}$ denote the unfinished blocks in~$G_j$. Then there exist nonoverlapping intervals~$I_1,\ldots,I_{\ell}\subseteq[1,n]$ with~$B_i\subseteq I_i$ for every~$i$.
\end{enumerate}
\end{lemma}
\begin{proof}
We prove the lemma by induction on~$j$. Let us first consider the base case~$j=1$. The consecutiveness property of~$T$ guarantees that the first column of~$C$ (which equals the first column of~$T$) contains at most two strictly positive entries. Let~$B$ denote the block that corresponds to these entries. The value of this block is one because the sum of all entries of the first column equals 1. If~$|B|=1$ then~$B$ is finished because if~$T$ contains only one positive entry in the first column, then this entry must be 1. If~$|B|=2$ then~$B$ is unfinished because neither of its rows is finished. In both cases the first statement of the lemma is true for block~$B$. All rows that have a zero in the first column form an unfinished block of their own with value zero. Also for these blocks the first statement is correct. The second statement is also correct because if~$|B|=1$ then the only difference between the blocks of~$G_0$ and~$G_1$ is that block~$B$ becomes finished, and if~$|B|=2$ then two unfinished blocks of~$G_0$ are merged. The correctness of the third statement follows from the fact that in the case~$|B|=2$, the two entries of~$B$ are consecutive due to the consecutiveness property of~$T$.

Now we come to the inductive step and assume that the statement is correct for the blocks of~$G_{j-1}$. Let~$I\subseteq[1,n]$ denote the set of indices~$i$ for which~$t_{ij}>0$. Observe that~$I$ can only be nondisjoint from unfinished blocks of~$G_{j-1}$. Due to the definition of~$G_j$ only blocks that are nondisjoint from~$I$ change from~$G_{j-1}$ to~$G_j$. Hence, the correctness of the first statement for all blocks of~$G_j$ that are disjoint from~$I$ follows from the induction hypothesis. If~$I$ is nondisjoint only from a single block~$B$ of~$G_{j-1}$ then this block will become finished. This follows from the fact that~$B$ has value~$|B|-1$ in~$G_{j-1}$ and that a total value of one is added to~$B$ because~$T$ is a doubly stochastic matrix. Hence, in this case~$G_{j-1}$ and~$G_j$ define the same set of blocks and the only difference is that~$B$ is unfinished in~$G_{j-1}$ and finished in~$G_j$. Then the correctness of all three statements follows from the induction hypothesis.

It remains to consider the case that~$I$ is nondisjoint from at least two blocks of~$G_{j-1}$. First we observe that~$I$ can be nondisjoint from at most two blocks of~$G_{j-1}$. Assume for contradiction that~$I$ is nondisjoint from three different blocks~$B_1$, $B_2$, and~$B_3$. Due to the third property, the induction hypothesis implies that one of these blocks must be entirely between the two others. Let~$B_2$ be this block. Since~$I$  is nondisjoint from~$B_1$ and~$B_3$, there are two indices~$i_1$ and~$i_3$ with~$t_{i_1j}>0$ and~$t_{i_3j}>0$ and~$i_1<i_2<i_3$ for all~$i_2\in B_2$. Due to the consecutiveness property of~$T$, this is only possible if all rows that belong to~$B_2$ are finished at column~$j$. Due to the induction hypothesis, the value of~$B_2$ at column~$j-1$ is~$|B_2|-1$. Hence, in order to finish all rows that belong to~$B_2$ one has to add a value of exactly 1 to~$B_2$ in column~$j$. Since column~$j$ of~$T$ sums to 1, this implies that there cannot be an index~$i\notin B_2$ with~$t_{ij}>0$, contradicting the choice of~$i_1$ and~$i_3$. This implies the correctness of the second property.

Hence, we only need to consider the case that~$I$ is nondisjoint from
exactly two blocks~$B_1$ and~$B_2$ of~$G_{j-1}$. Due to the induction
hypothesis the values of these blocks at column~$j-1$ are~$|B_1|-1$
and~$|B_2|-1$, respectively. Since column~$j$ of~$T$ has a sum of 1, the value of the block~$B$ in~$G_j$ that emerges from merging~$B_1$ and~$B_2$ has a value of~$(|B_1|-1)+(|B_2|-1)+1=|B_1|+|B_2|-1=|B|-1$. This proves the first property. To prove the third property, we use the fact that the consecutiveness property of~$T$ guarantees that there cannot be an unfinished block between~$B_1$ and~$B_2$ in~$G_{j-1}$. Hence, we can associate with~$B$ the smallest interval that contains the intervals~$I_1$ and~$I_2$ that were associated with~$B_1$ and~$B_2$ in~$G_{j-1}$. This also proves the third property.
\end{proof}

One might ask if the consecutiveness property is satisfied by every
optimal extreme point of the linear program. Let us mention that this
is not the case. A simple counterexample is provided by the
instance~$X=\{9,6,4,1\}$ and~$Y=\{5,5,5,5\}$. In this instance, an
optimal extreme point would be, for example, to take one half of each
of the items~$x_1$ and~$x_4$ in steps one and three and to take one
half of each of the items~$x_2$ and~$x_3$ in steps two and four. This
extreme point does not, however, satisfy the consecutiveness
property. Hence, the transformation described in this section is
necessary.

\subsection{Rounding}

In this section, we use the transformed matrix $T$ to create the
solution matrix~$R$, which is a doubly stochastic $0/1$ matrix,
i.e., a permutation matrix. We apply the following rounding method.

  \begin{algorithmic}[1]
  	 \For{$j=1$ to $n$}
  	   \State Let $B$ denote the \emph{active} block in $G_j$, i.e., the block that contains the rows~$i$ with~$t_{ij}>0$.
      \State Let~$p$ denote the smallest index in $B$ such that $r_{pi} = 0$ for all
  $i<j$. \label{line:Rounding}
       \State Set $r_{pj} = 1$ and~$r_{qj}=0$ for all~$q\neq p$.
  	 \EndFor
  \end{algorithmic}

Observe that the first step is well-defined because all nonzero
entries in column~$j$ belong by definition to the same block of~$G_j$.
The resulting matrix $R$ will be doubly stochastic, since each column
contains a single one, as does each row. We just need to prove that in
Line~\ref{line:Rounding} there always exists a row~$p\in B$ that is
unfinished in~$R$ at column~$j-1$.  This follows from the first part
of the next lemma because, due to Lemma~\ref{lem:BlockStructure}, the
active block~$B$ in~$G_j$ emerges from one or two unfinished blocks
in~$G_{j-1}$ and these blocks each contain a row that is unfinished
in~$R$ at column~$j-1$.

\begin{lemma}\label{lem:round}
Let $B$ be a block in $G_j$ for some~$j\in\{1,\ldots,n\}$.
\begin{enumerate}
\item If $B$ is an unfinished block in $G_j$ and $p$ is the largest index in $B$, then
$r_{pi} = 0$ for all $i\leq j$ and all rows corresponding to
$B\setminus\{p\}$ are finished in $R$ at column $j$.
\item If $B$ is a finished block in $G_j$, then for all $q \in B$, row
  $q$ is finished in $R$ at column $j$.
\end{enumerate}
\end{lemma}
\begin{proof}
We will prove the lemma by induction on~$j$.
Let us first consider the base case~$j=1$. The consecutiveness property of~$T$ guarantees that the first column of~$T$ contains at most two strictly positive entries. Let~$B$ denote the block that corresponds to these entries. If~$|B|=1$ then~$B=\{p\}$ is finished in~$T$ at column~$1$ and the rounding will set~$r_{p1}=1$.
If~$|B|=2$ then~$B=\{p,q\}$ is unfinished and the rounding will set~$r_{p1}=1$ if~$p<q$. In both cases the statement of the lemma is correct for~$B$. All other blocks in~$G_1$ are unfinished singleton blocks, for which the lemma is also true.

Now let us assume that the lemma is true for $j-1$ and prove it for $j$.
By property 2 of Lemma~\ref{lem:BlockStructure}, the
blocks in $G_j$ emerge from the blocks in $G_{j-1}$ either by
merging exactly two unfinished blocks or making one unfinished block
finished.  In the former case, suppose we merge two blocks $B_1$ and~$B_{2}$.  Let $\ell_1$ and $\ell_2$ denote the largest indices in
$B_1$ and $B_{2}$, respectively, and assume that~$\ell_1<\ell_2$.  By
assumption, we have that
$r_{\ell_1 i}=r_{\ell_2 i}=0$ for all
$i \le j-1$.  Thus, we can set $r_{\ell_1 j} = 1$, and the first statement
will still hold for the new unfinished block in $G_j$.
In the latter case, suppose that~$B$ is an unfinished block in
$G_{j-1}$ that becomes finished in $G_j$ and that $\ell$ is the
largest label in~$B$.  Then by assumption, $r_{\ell i} = 0$ for all
$i \leq j-1$, so we can set $r_{\ell j} = 1$ and statement (ii) holds.
\end{proof}

We define the \emph{value} of a permutation matrix~$M$ to be the smallest~$\gamma$ for which there exist~$\alpha'$ and~$\beta'$ with~$\gamma=\beta'-\alpha'$ such that~$(M,\alpha',\beta')$ is a feasible solution to the linear program.
\begin{theorem}\label{thm:LPRounding}
Let~$(T,\alpha,\beta)$ be an optimal solution to the linear program.
Then $(R,\alpha,\beta+\mu_x)$ is a feasible solution to the linear program.
Hence, the value of the matrix~$R$ is at most $(\beta-\alpha) + \mu_{x} \leq 2\cdot \text{OPT}$,
where~$\text{OPT}$ denotes the value of the optimal permutation matrix.
\end{theorem}
For ease of notation, we define $r_j$ as follows:
$r_j  =  \sum_{i=1}^n r_{ij} \cdot x_i$.
Note that~$r_j$ corresponds to the value of the element from $X$ that
the algorithm places in position~$j$.
We will see later that Theorem~\ref{thm:LPRounding} follows easily from the next lemma.
\begin{lemma}\label{lemma:ErrorPrefix}
For each~$k\in\{1,\ldots,n\}$,
\vspace{-2mm}
\begin{eqnarray}
\sum_{j=1}^{k} (r_j-t_j) \in [0,\mu_{x}].
\end{eqnarray}
\end{lemma}

We need the following lemma in the proof of Lemma~\ref{lemma:ErrorPrefix}.
\begin{lemma}\label{lem:excess}
Let~$b$ be the largest index in an unfinished block~$B$ in~$G_j$.
Then,
\begin{eqnarray*}
c_{bj} & = & \sum_{i \in B\setminus\{b\}} (1-c_{ij}).
\end{eqnarray*}
\end{lemma}
\begin{proof}
Let the value of the unfinished block~$B$ be~$k=\sum_{i \in B} c_{ij}$.
By property 1 of Lemma \ref{lem:BlockStructure}, block~$B$ consists of $k+1$ rows.
 Thus, we have
\[
c_{bj}  =  k - \sum_{i \in B\setminus\{b\}} c_{ij}
 =  \sum_{i \in B\setminus\{b\}} (1 - c_{ij}).\qedhere
\]
\end{proof}

\begin{proof}[Proof of Lemma~\ref{lemma:ErrorPrefix}]
\newcommand{\error}{\text{er}}
Let us consider the sets of finished and unfinished blocks in $G_{{k}}$, $\mathcal{B}_F$ and $\mathcal{B}_U$, respectively.
For a block~$B\in\mathcal{B}_F\cup\mathcal{B}_U$, we denote by
\[ 
   \error_{{k}}(B) = \sum_{i\in B}\sum_{j=1}^{{k}} x_i(r_{ij}-t_{ij})
\]
its rounding error. Since each row is contained in exactly one block of~$G_{{k}}$,
\begin{equation}
    \sum_{j=1}^{{k}} (r_j-t_j) = \sum_{j=1}^{{k}} \sum_{i=1}^n x_i (r_{ij}-t_{ij})
    =  \sum_{i=1}^n \sum_{j=1}^{{k}} x_i(r_{ij}-t_{ij})
    = \sum_{B\in\mathcal{B}_F\cup\mathcal{B}_U} \error_{{k}}(B).\label{eqn:Error1}
\end{equation}
Hence, in order to prove the lemma, it suffices to bound the rounding errors of the blocks.

If block~$B$ is finished in~$G_{{k}}$, then all rows that belong to~$B$ are finished in~$T$ and in~$R$ (due to property~2 of Lemma~\ref{lem:round}) at column~${k}$. Hence,
\begin{equation}
   \error_{{k}}(B) = \sum_{i\in B}\sum_{j=1}^{{k}} x_i(r_{ij}-t_{ij}) 
                    = \sum_{i\in B}x_i\cdot \Bigg(\sum_{j=1}^{{k}} r_{ij}-\sum_{j=1}^{{k}}t_{ij}\Bigg) 
                    = \sum_{i\in B}x_i\cdot (1-1) 
                    = 0.\label{eqn:Error2}
\end{equation}

Now consider an unfinished block~$B$ in~$G_{{k}}$, and let~$a$ and~$b$ denote the smallest and largest index in~$B$, respectively.
By Lemma \ref{lem:round}, all rows in the block except for~$b$ are finished in $R$ at column~${k}$ (i.e., $\sum_{j=1}^{{k}}r_{ij}=1$ for~$i\in B\setminus\{b\}$ and $\sum_{j=1}^{{k}}r_{bj}=0$).
The rounding error of~$B$ can thus be bounded as follows (remember that~$c_{i{k}}=\sum_{j=1}^{{k}}t_{ij}$):
\begin{align}
  \error_{{k}}(B) & = \sum_{i\in B}\sum_{j=1}^{{k}} x_i(r_{ij}-t_{ij}) 
                     = \sum_{i\in B} x_i\sum_{j=1}^{{k}}r_{ij} - \sum_{i\in B} x_i \sum_{j=1}^{{k}}t_{ij}\notag\\
                   & = \sum_{i\in B\setminus\{b\}} x_i - \sum_{i\in B} x_i c_{i{k}}
                     = \sum_{i\in B\setminus\{b\}} x_i(1-c_{i{k}}) - x_b c_{b{k}}\notag\\
                   & = \sum_{i\in B\setminus\{b\}} x_i(1-c_{i{k}}) - x_b \sum_{i\in B\setminus\{b\}}(1-c_{i{k}})\label{eq:usage1}\\
                   & = \sum_{i\in B\setminus\{b\}} (x_i-x_b)(1-c_{i{k}})\notag\\                   
                   & \le (x_a-x_b)\sum_{i\in B\setminus\{b\}} (1-c_{i{k}})\label{eq:usage4}\\
                   & = (x_a-x_b)\cdot c_{bj}\label{eq:usage3}\\
                   & \le x_a-x_b.\label{eq:usage2}
\end{align}
Equations~\eqref{eq:usage1} and~\eqref{eq:usage3} follow from Lemma \ref{lem:excess}.
Inequality \eqref{eq:usage2} follows from the fact that $c_{bj} \le 1$. Inequality~\eqref{eq:usage4} follows from 
the facts that $1-c_{i{k}}\ge 0$ and~$x_i-x_b\ge 0$ for all~$i\in B$. These facts also imply that~$\error_{{k}}(B)\ge 0$.
Hence,
\begin{equation}
   \error_{{k}}(B) \in [0,x_a-x_b].\label{eqn:Error3}
\end{equation}

Together~\eqref{eqn:Error1} and~\eqref{eqn:Error2} imply that
\begin{align}\label{eqn:TotalError}
    \sum_{j=1}^{{k}} (r_j-t_j) & = \sum_{B\in\mathcal{B}_F\cup\mathcal{B}_U} \error_{{k}}(B)
    = \sum_{B\in\mathcal{B}_F} \error_{{k}}(B) + \sum_{B\in\mathcal{B}_U} \error_{{k}}(B)
    =  \sum_{B\in\mathcal{B}_U} \error_{{k}}(B).
\end{align}

Now, let $B_1, \dots B_h$ denote the unfinished blocks in $G_{{k}}$,
and for each block $B_f$ in $\mathcal{B}_U$, let $a_f$ and $b_f$ denote the minimum and
maximum indices, respectively, contained in the block. Property~3 of Lemma~\ref{lem:BlockStructure} implies
that the intervals~$[a_f,b_f]$ are pairwise disjoint. Hence,
\eqref{eqn:Error3} implies that
\[
   \sum_{B\in\mathcal{B}_U} \error_{{k}}(B) \in \Bigg[ 0, \sum_{f=1}^{h}(x_{a_f}-x_{b_f})\Bigg] \subseteq \Big[0,x_{1}-x_{n}\Big] \subseteq [0,\mu_x]. 
\]
Together with~\eqref{eqn:TotalError}, this implies the lemma.
\end{proof}

Now we are ready to prove Theorem~\ref{thm:LPRounding}.

\begin{proof}[Proof of Theorem~\ref{thm:LPRounding}]
Let~$(T,\alpha,\beta)$ denote an optimal solution to the linear program.
By definition, our rounding method produces a permutation matrix~$R$. Lemma~\ref{lemma:ErrorPrefix} implies that~$(R,\alpha,\beta+\mu_{y})$ is also a feasible solution to the linear program because for each~$k\in\{1,\ldots,n\}$,
\[
  \sum_{j=1}^{k}\sum_{i=1}^n r_{ij}\cdot x_i - \sum_{j=1}^{k-1} y_j
= \sum_{j=1}^{k}r_j - \sum_{j=1}^{k-1} y_j
\le \sum_{j=1}^{k}t_j - \sum_{j=1}^{k-1} y_j  + \mu_x
\le \beta + \mu_x
\]
and
\[
  \sum_{j=1}^{k}\sum_{i=1}^n r_{ij}\cdot x_i - \sum_{j=1}^k y_j
= \sum_{j=1}^{k}r_j - \sum_{j=1}^k y_j
\ge \sum_{j=1}^{k}t_j - \sum_{j=1}^k y_j
\ge \alpha.
\]

Now the theorem follows because~$\text{OPT}\ge \mu_x$ and~$\text{OPT}\ge \beta-\alpha$.
\end{proof}

Finally, we note that the additive integrality gap of the linear program in
Section \ref{section:GasolineProblem} can be arbitrarily close to
$\mu_x$.  Consider the following instance:
\begin{eqnarray*}
y = \frac{(n-1) + \mu}{n}, \quad Y ~=~ \{\underbrace{y, \dots,
  y}_{n \text{ entries}}\}, \quad X ~ =~ \{\mu, \underbrace{1, 1, \dots,
  1}_{n-1 \text{ entries }}\}.
\end{eqnarray*}
Then the value of the linear program is $y$.  However, the optimal
value is $\mu$, which can be arbitrarily larger than $y$.

\subsection{LP Rounding for the Slated Stock Size Problem}

We show that Theorem~\ref{thm:LPRounding} can also be applied to the
slated stock size problem, defined in
Section~\ref{subsec:SlatedSSProblem}.
Let $X =\{x_1\ge\ldots \ge x_{n_x}\}$ and $Y = \{y_1 \ge \ldots \ge
y_{n_y}\}$ be an input for the slated stock size problem, and
let~$\mu_x=x_1$, $\mu_y=y_1$.  Recall that in this problem,
arbitrary disjoint subsets of $n_x$ and $n_y$ slots are slated for $x$- and
$y$-jobs, respectively. Remember that $I_x$ and $I_y$ denote the indices of the
$x$- and $y$-slots, respectively. 
Let $\eta$ denote the optimal value for
the relaxation of the following integer program:
\begin{align}
 &\min \beta-\alpha \notag\\
&\forall j \in I_x: \sum_{i=1}^{n_x} z_{ij}^x  = 1, \quad
\forall i \in [1,n_x]: \sum_{j\in I_x} z_{ij}^x = 1, \quad \forall i,j \in I_x: z_{ij}^x  \in \{0,1\},\notag\\
&\forall j \in I_y: \sum_{i=1}^{n_y} z_{ij}^y  = 1, \quad
\forall i \in [1,n_y]: \sum_{j\in I_y} z_{ij}^y = 1, \quad \forall i,j \in I_y: z_{ij}^y  \in \{0,1\},\notag\\
&\forall \text{~prefix~$P$} : \quad 
\alpha ~\leq~ \sum_{j\in P\cap I_x}\sum_{i=1}^{n_x} z_{ij}^x\cdot  x_i -
\sum_{j\in P\cap I_y}\sum_{i=1}^{n_y} z_{ij}^y\cdot  y_i 
 ~\leq~ \beta. \notag
\end{align}

Now let~$z_{ij}^x$ and~$z_{ij}^y$ denote an optimal solution to the
relaxation of the previous integer program.  Consider the generalized
gasoline problem with input~$\tilde{x}_{j} = \sum_{i=1}^{n_x}
z_{ij}^x\cdot x_i$ for each $x$-slot~$j\in I_x$
and~$y_1,\ldots,y_{n_y}$ to be assigned to the $y$-slots. That is, in
the constructed problem instance of the generalized gasoline problem,
the $x$-values are fixed for each slot while the $y$-values are to be
permuted.  The optimal fractional solution for this instance still has
value $\eta$.  Hence, Theorem~\ref{thm:LPRounding} implies that we
obtain a permutation~$\pi$ of the items~$y_1,\ldots,y_{n_y}$ with
value at most~$\eta+\mu_y$. Now we change the roles of~$x$ and~$y$ and
consider the generalized gasoline problem with
input~$y_{\pi(1)},\ldots,y_{\pi(n_y)}$ (these are the fixed items in
the $y$-slots) and~$x_j$ for~$j\in I_x$ (these items are to be
permuted). The optimal fractional solution for this instance has value
at most~$\eta+\mu_y$.  Hence, Theorem~\ref{thm:LPRounding} implies
that we obtain a permutation~$\sigma$ of the items~$x_j$ with value at
most~$\eta+\mu_y+\mu_x$. The permutations~$\pi$ and~$\sigma$ together
form a solution for the slated stock size problem with value at
most~$\eta+\mu_y+\mu_x\le 3~OPT$.


\section{Conclusions}

We have introduced two new variants of the stock size problem and have
presented nontrivial approximation algorithms for them. The most
intriguing question for our variants as well as for the original stock
size problem is if the approximation guarantees can be improved. Each
of these problems is NP-hard but no APX-hardness is known. So it is
conceivable that there exists a PTAS. Closing this gap seems very
challenging.

\subsubsection*{Acknowledgments.} 
We would like to thank Anupam Gupta for
several useful examples and enlightening discussions.
We would like to thank Jochen K\"onemann for pointing out a connection
between the alternating stock size problem and the optimization
version of the gasoline puzzle.  We would also like to thank Tam\'as
Kis for pointing out related papers on scheduling problems.

Part of this work was done during the trimester program on
  Combinatorial Optimization at the Hausdorff Institute for
  Mathematics in Bonn and we would like to thank HIM for their
  organization and hospitality during the program.
\bibliography{stockSize}


\appendix

\section{NP-hardness} \label{sec:np}

The goal of this section is to prove the following theorem.  
\begin{theorem}\label{np-hard}
	The alternating stock size problem is NP-hard.
\end{theorem}  

We give a reduction from the 3-partition problem, which is defined as follows.

  \indent\indent {\bf Input:} $ Z = \{ z_1, ..., z_n\}$ where $n = 3k$; $z_i \in (\frac{1}{4},\frac{1}{2})$, $\forall i \in \{1,...,n\}$, and $\sum_{i=1}^n z_i = k$.
  
  \indent\indent {\bf Question:} Can we divide the input set $Z$ into sets of three elements each such that each set has value exactly 1?

\begin{proof}[Proof of Theorem~\ref{np-hard}]

We will reduce 3-partition to the alternating stock size problem.
Consider an instance $Z$ of $3$-partition.  We use $Z$ to create the
following instance $I$ of the alternating stock size problem.
    
  \indent\indent {\bf Input:} $X=\{x_1,...,x_{n+k}\}$ where $x_i=1, ~\forall i \in \{1,...,n+k\}$ and $Y = \{y_1,...,y_{n+k}\}$ where $y_i = 1 - z_i,~ \forall i \in \{1,...,n\} $ and $y_i = 2,~ \forall i \in \{n+1,...,n+k\}$.
  
  \indent\indent {\bf Question:} Is there a solution for the alternating stock size problem where the maximum stock size is at most two?

  \vspace{1mm}
  
  We will show that there is a feasible 3-partition for $Z$ iff there is a solution for the alternating stock size problem where the maximum stock size is at most two.

  Suppose that there is a 3-partition for $Z$. We want to show that we can find an alternating sequence for $I$ with a stock size of at most 2.  We can assume that each set of the partition contains $\{z_{3i+1},z_{3i+2},z_{3i+3}\}$ with $ i \in \{0,...,k-1\}$ (if not, change the indices of the elements of $Z$). 
  We have $z_{1} + z_{2} + z_{3} = 1$ so the sequence $1,-y_1,1,-y_2,1,-y_3,1,-2$ never exceeds two and at the end of the sequence the sum is back at zero. So we repeat this process to the other sets of the 3-partition.
  
\begin{figure}[ht!]
  \centering
  \scalebox{.4}{
\ifx\du\undefined
  \newlength{\du}
\fi
\setlength{\du}{15\unitlength}
\begin{tikzpicture}
\pgftransformxscale{1.000000}
\pgftransformyscale{-1.000000}
\definecolor{dialinecolor}{rgb}{0.000000, 0.000000, 0.000000}
\pgfsetstrokecolor{dialinecolor}
\definecolor{dialinecolor}{rgb}{1.000000, 1.000000, 1.000000}
\pgfsetfillcolor{dialinecolor}
\definecolor{dialinecolor}{rgb}{0.921569, 0.000000, 0.000000}
\pgfsetstrokecolor{dialinecolor}
\node[anchor=west] at (-32.691000\du,-55.524700\du){};
\pgfsetlinewidth{0.150000\du}
\pgfsetdash{}{0pt}
\pgfsetdash{}{0pt}
\pgfsetbuttcap
{
\definecolor{dialinecolor}{rgb}{0.000000, 0.000000, 0.000000}
\pgfsetfillcolor{dialinecolor}
\definecolor{dialinecolor}{rgb}{0.000000, 0.000000, 0.000000}
\pgfsetstrokecolor{dialinecolor}
\draw (-45.000000\du,-50.000000\du)--(-45.000000\du,-50.000000\du);
}
\pgfsetlinewidth{0.200000\du}
\pgfsetdash{}{0pt}
\pgfsetdash{}{0pt}
\pgfsetbuttcap
{
\definecolor{dialinecolor}{rgb}{0.000000, 0.000000, 0.000000}
\pgfsetfillcolor{dialinecolor}
\pgfsetarrowsend{to}
\definecolor{dialinecolor}{rgb}{0.000000, 0.000000, 0.000000}
\pgfsetstrokecolor{dialinecolor}
\draw (-45.870300\du,-50.000000\du)--(-18.000000\du,-50.000000\du);
}
\pgfsetlinewidth{0.200000\du}
\pgfsetdash{}{0pt}
\pgfsetdash{}{0pt}
\pgfsetbuttcap
{
\definecolor{dialinecolor}{rgb}{0.000000, 0.000000, 0.000000}
\pgfsetfillcolor{dialinecolor}
\pgfsetarrowsend{to}
\definecolor{dialinecolor}{rgb}{0.000000, 0.000000, 0.000000}
\pgfsetstrokecolor{dialinecolor}
\draw (-45.000000\du,-49.000000\du)--(-45.000000\du,-68.000000\du);
}
\pgfsetlinewidth{0.200000\du}
\pgfsetdash{}{0pt}
\pgfsetdash{}{0pt}
\pgfsetroundjoin
\pgfsetbuttcap
{
\definecolor{dialinecolor}{rgb}{0.074510, 0.023529, 0.843137}
\pgfsetfillcolor{dialinecolor}
{\pgfsetcornersarced{\pgfpoint{0.010000\du}{0.010000\du}}\definecolor{dialinecolor}{rgb}{0.074510, 0.023529, 0.843137}
\pgfsetstrokecolor{dialinecolor}
\draw (-45.000000\du,-50.000000\du)--(-41.000000\du,-58.000000\du)--(-39.000000\du,-53.000000\du)--(-36.000000\du,-61.000000\du)--(-33.000000\du,-55.000000\du)--(-30.000000\du,-63.000000\du)--(-27.000000\du,-57.000000\du)--(-24.000000\du,-65.000000\du)--(-21.000000\du,-50.000000\du);
}}
\pgfsetlinewidth{0.150000\du}
\pgfsetdash{{1.000000\du}{1.000000\du}}{0\du}
\pgfsetdash{{1.000000\du}{1.000000\du}}{0\du}
\pgfsetbuttcap
{
\definecolor{dialinecolor}{rgb}{1.000000, 0.027451, 0.027451}
\pgfsetfillcolor{dialinecolor}
}
\definecolor{dialinecolor}{rgb}{1.000000, 0.027451, 0.027451}
\pgfsetstrokecolor{dialinecolor}
\draw (-39.000000\du,-52.332295\du)--(-39.000000\du,-50.667705\du);
\pgfsetlinewidth{0.150000\du}
\pgfsetdash{}{0pt}
\pgfsetmiterjoin
\definecolor{dialinecolor}{rgb}{1.000000, 0.027451, 0.027451}
\pgfsetstrokecolor{dialinecolor}
\draw (-38.750000\du,-52.332295\du)--(-39.000000\du,-52.832295\du)--(-39.250000\du,-52.332295\du)--cycle;
\pgfsetlinewidth{0.150000\du}
\pgfsetdash{}{0pt}
\pgfsetmiterjoin
\definecolor{dialinecolor}{rgb}{1.000000, 0.027451, 0.027451}
\pgfsetstrokecolor{dialinecolor}
\draw (-39.250000\du,-50.667705\du)--(-39.000000\du,-50.167705\du)--(-38.750000\du,-50.667705\du)--cycle;
\pgfsetlinewidth{0.150000\du}
\pgfsetdash{{1.000000\du}{1.000000\du}}{0\du}
\pgfsetdash{{1.000000\du}{1.000000\du}}{0\du}
\pgfsetbuttcap
{
\definecolor{dialinecolor}{rgb}{1.000000, 0.027451, 0.027451}
\pgfsetfillcolor{dialinecolor}
}
\definecolor{dialinecolor}{rgb}{1.000000, 0.027451, 0.027451}
\pgfsetstrokecolor{dialinecolor}
\draw (-33.000000\du,-53.667705\du)--(-33.000000\du,-54.332295\du);
\pgfsetlinewidth{0.150000\du}
\pgfsetdash{}{0pt}
\pgfsetmiterjoin
\definecolor{dialinecolor}{rgb}{1.000000, 0.027451, 0.027451}
\pgfsetstrokecolor{dialinecolor}
\draw (-33.250000\du,-53.667705\du)--(-33.000000\du,-53.167705\du)--(-32.750000\du,-53.667705\du)--cycle;
\pgfsetlinewidth{0.150000\du}
\pgfsetdash{}{0pt}
\pgfsetmiterjoin
\definecolor{dialinecolor}{rgb}{1.000000, 0.027451, 0.027451}
\pgfsetstrokecolor{dialinecolor}
\draw (-32.750000\du,-54.332295\du)--(-33.000000\du,-54.832295\du)--(-33.250000\du,-54.332295\du)--cycle;
\pgfsetlinewidth{0.150000\du}
\pgfsetdash{{1.000000\du}{1.000000\du}}{0\du}
\pgfsetdash{{1.000000\du}{1.000000\du}}{0\du}
\pgfsetbuttcap
{
\definecolor{dialinecolor}{rgb}{1.000000, 0.027451, 0.027451}
\pgfsetfillcolor{dialinecolor}
}
\definecolor{dialinecolor}{rgb}{1.000000, 0.027451, 0.027451}
\pgfsetstrokecolor{dialinecolor}
\draw (-27.000000\du,-56.332295\du)--(-27.000000\du,-55.667705\du);
\pgfsetlinewidth{0.150000\du}
\pgfsetdash{}{0pt}
\pgfsetmiterjoin
\definecolor{dialinecolor}{rgb}{1.000000, 0.027451, 0.027451}
\pgfsetstrokecolor{dialinecolor}
\draw (-26.750000\du,-56.332295\du)--(-27.000000\du,-56.832295\du)--(-27.250000\du,-56.332295\du)--cycle;
\pgfsetlinewidth{0.150000\du}
\pgfsetdash{}{0pt}
\pgfsetmiterjoin
\definecolor{dialinecolor}{rgb}{1.000000, 0.027451, 0.027451}
\pgfsetstrokecolor{dialinecolor}
\draw (-27.250000\du,-55.667705\du)--(-27.000000\du,-55.167705\du)--(-26.750000\du,-55.667705\du)--cycle;
\definecolor{dialinecolor}{rgb}{0.000000, 0.000000, 0.000000}
\pgfsetstrokecolor{dialinecolor}
\node[anchor=west] at (-38.494964\du,-51.469682\du){\scalebox{2}{\textcolor{red}{$z_1$}}};
\definecolor{dialinecolor}{rgb}{0.000000, 0.000000, 0.000000}
\pgfsetstrokecolor{dialinecolor}
\node[anchor=west] at (-32.497619\du,-53.807156\du){\scalebox{2}{\textcolor{red}{$z_2$}}};
\definecolor{dialinecolor}{rgb}{0.000000, 0.000000, 0.000000}
\pgfsetstrokecolor{dialinecolor}
\node[anchor=west] at (-26.487345\du,-55.751656\du){\scalebox{2}{\textcolor{red}{$z_3$}}};
\pgfsetlinewidth{0.150000\du}
\pgfsetdash{{1.000000\du}{1.000000\du}}{0\du}
\pgfsetdash{{1.000000\du}{1.000000\du}}{0\du}
\pgfsetbuttcap
{
\definecolor{dialinecolor}{rgb}{1.000000, 0.027451, 0.027451}
\pgfsetfillcolor{dialinecolor}
\definecolor{dialinecolor}{rgb}{1.000000, 0.027451, 0.027451}
\pgfsetstrokecolor{dialinecolor}
\draw (-46.000000\du,-65.000000\du)--(-21.000000\du,-65.000000\du);
}
\definecolor{dialinecolor}{rgb}{1.000000, 0.027451, 0.027451}
\pgfsetstrokecolor{dialinecolor}
\node[anchor=west] at (-47.300000\du,-65.0006\du){\scalebox{2}{\textcolor{red}{2}}};
\pgfsetlinewidth{0.130000\du}
\pgfsetdash{{\pgflinewidth}{0.200000\du}}{0cm}
\pgfsetdash{{\pgflinewidth}{0.200000\du}}{0cm}
\pgfsetbuttcap
{
\definecolor{dialinecolor}{rgb}{1.000000, 0.027451, 0.027451}
\pgfsetfillcolor{dialinecolor}
\definecolor{dialinecolor}{rgb}{1.000000, 0.027451, 0.027451}
\pgfsetstrokecolor{dialinecolor}
\draw (-33.500000\du,-55.000000\du)--(-26.500000\du,-55.000000\du);
}
\pgfsetlinewidth{0.130000\du}
\pgfsetdash{{\pgflinewidth}{0.200000\du}}{0cm}
\pgfsetdash{{\pgflinewidth}{0.200000\du}}{0cm}
\pgfsetbuttcap
{
\definecolor{dialinecolor}{rgb}{1.000000, 0.027451, 0.027451}
\pgfsetfillcolor{dialinecolor}
\definecolor{dialinecolor}{rgb}{1.000000, 0.027451, 0.027451}
\pgfsetstrokecolor{dialinecolor}
\draw (-39.500000\du,-53.000000\du)--(-32.500000\du,-53.000000\du);
}
\end{tikzpicture}}
  \caption{The sequence $1,y_1,1,y_2,1,y_3,1,2$}\label{np2}
\end{figure}
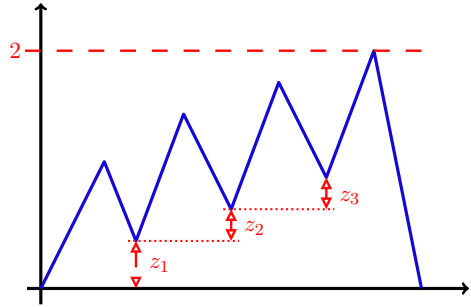

  Suppose that there is a solution for the alternating stock size problem such that $OPT(I) \leq 2$. We want to find a 3-partition for the instance $Z$. In $Y$ some elements are less than $1$ and can be written $1-z$ for $z \in Z$ and others are equal to $2$. The idea is to show that between any two consecutive $2$'s, there are always exactly three $y$'s whose values are less than $1$.  

Let $x_1^*,y_1^*,\dots,x_{n+k}^*,y_{n+k}^*$ be the sequence for the optimal solution, and let $y_j^*$ be the first element in the sequence that is a $2$. Then $\sum_{i=1}^jx_i^* -\sum_{i=1}^{j-1}y_i^* = 2$ because the optimal is at most two and the sum should not go below zero at the next step. Moreover $y_j^*$ is the first $2$ so: $\forall i \in \{1,...,j-1\},~ \exists z_i^* \in Z,~ y_j^* = 1-z_i^*$. Therefore $\sum_{i=1}^{j-1}z_i^* = 1$.  And $\forall i \in
  \{1,...,n\},~ z_i \in (\frac{1}{4},\frac{1}{2})$ so $j=4$ and $z_{1}^* + z_{2}^* + z_{3}^* = 1$.  Thus, we have packed three elements of $Z$ and after $y_j^*$ the sum is back to zero, so we can repeat the process on the remaining elements.
  
\end{proof}

From this proof, we can see that the Gasoline Problem is also NP-hard.
This follows from the fact that, in
the reduction for the alternating stock size problem,
all of the $x$-values are set equal to one.
Thus, they can simply be fixed in advance.  Then, the only decisions
required in the problem produced in the reduction involve placing the
$y$-values.  More specifically, one can see that the NP-hardness proof
also shows that the problem of placing the $y$-values so as to
minimize the difference between computing the highest point and
the lowest point is NP-hard.


\section{Proof of Lemma~\ref{thm:batch_thm}}\label{appendix2}


\begin{proof}[Proof of Lemma~\ref{thm:batch_thm}]
For a pair $B = \{x,y\} \in X \times Y$, let $x(B)$ and $y(B)$ denote
the values of the $x$- and $y$-jobs, respectively.  We will sometimes
refer to a pair $\{x,y\}$ as positive or negative which describes
the value of $x-y$.

Partition the pairs into two sets ${\mathscr B}^+$
and ${\mathscr B}^-$, where the first set contains all of the
positive pairs, and the second
set contains all of the negative pairs.
(We can assume there are no pairs for which $x-y=0$, since these can
simply be sequenced first.)
Begin with an empty list $L^S$ and with the current stock size $S$ set
to zero.  
We
then repeat the following step until the sets ${\mathscr B}^+$ and
${\mathscr B}^-$ are both empty:

``Find any pair $\{x,y\}$ in ${\mathscr B}^-$ such that $S + x - y
\geq 0$.  If no such pair exists, choose a pair $\{x,y\}$ from
${\mathscr B}^+$.  Set $S := S + x - y$.  Append $x$ and then $y$ to $L^S$
and remove the pair from ${\mathscr B}^-$ or from ${\mathscr B}^+$.''


Since the sum of all the pairs is zero, and the stock size never goes
below zero, each time a positive pair is appended to the list, there
exists at least one negative pair in ${\mathscr B}^-$.  To prove the
upper bound $(1+q)T$ on the maximum stock size, we will introduce
so-called {\em breakpoints} which fulfill the following two
conditions: (a) at each breakpoint, the current stock size $S$ is less
than $qT$; and (b) between any two consecutive breakpoints, $S$
remains below $(1+q)T$.

The first breakpoint is at time zero; the other breakpoints are the
time points just before a positive pair is sequenced.  The last
breakpoint is defined to be just after the last pair is sequenced.
The first breakpoint and the last one fulfill condition (a) by
definition.  If one of the other breakpoints would not fulfill
condition (a), then $S \geq qT$ must hold, and because of property
(ii) our algorithm would have chosen a negative pair to be sequenced.

Next we consider two consecutive breakpoints $BP_i$ and $BP_{i+1}$,
and we let $S_i < qT$ denote the stock size at time $BP_i$.  Since at
time $BP_i$, a positive pair is sequenced, for all negative pairs
$B^-$, the inequality
$$S_i + x(B^-) < T,$$ is true.  Otherwise, a negative pair $B^-$
could have been sequenced next.

After positive pair $B^+$ is appended to $L^S$, the current stock size
increases to $S_i + x(B^+) - y(B^+) \leq S_i + qT$.  Clearly $S_i +
x(B^+) \leq qT + T$.  Either the next pair is positive (and we are
again at a breakpoint) or there is a sequence of negative pairs.
The stock size during any of these pairs always remains below $S_i +
x(B^+) - y(B^+) + x(B^-) < T + qT$.  After each of these pairs, the
stock size does not increase.  This shows that condition (b) holds for
any two consecutive breakpoints and completes the proof.
\end{proof}

\end{document}